\documentclass[sigconf, nonacm]{acmart}
\usepackage{amsmath}
\allowdisplaybreaks
\usepackage{booktabs}
\usepackage{caption}
\usepackage{graphicx}
\usepackage[export]{adjustbox}
\usepackage{diagbox}
\usepackage{xspace}
\usepackage{algpseudocodex}
\usepackage{bm}
\usepackage[utf8]{inputenc}
\usepackage{tikz}
\usepackage{pgfplots}
\usepackage{amsthm}
\usepackage{multicol}
\usepackage{hyperref}
\usepackage{color}
\usepackage{algorithm}
\usepackage{algorithmicx}  
\usepackage{relsize}
\hypersetup{colorlinks,linkcolor={violet},citecolor={violet},urlcolor={black}}  
\usepackage{natbib}
\setcitestyle{square, numbers, sort}
\usepackage[capitalize]{cleveref}
\crefformat{section}{\S#2#1#3}
\crefmultiformat{section}{\S#2#1#3}{ and~\S#2#1#3}{, \S#2#1#3}{, and~\S#2#1#3}
\crefformat{subsection}{\S#2#1#3}
\crefmultiformat{subsection}{\S#2#1#3}{ and~\S#2#1#3}{, \S#2#1#3}{, and~\S#2#1#3}

\crefname{algocf}{Alg.}{Algs.} 
\Crefname{algocf}{Alg.}{Algs.}
\crefname{algorithm}{Alg.}{Algs.}
\Crefname{algorithm}{Algorithm}{Algorithms}
\crefname{ALC@unique}{alg.}{algs.}
\Crefname{ALC@unique}{Algorithm}{Algorithms}
\usepackage{subcaption}
\usepackage{float}

\usepackage{enumitem}
\setlist[description]{leftmargin=1em}
\usepackage{balance}

\usepackage{amsbsy}
\usepackage{listings}
\usepackage{tikzscale}
\usepackage{xcolor}

\newlength{\setupLength}
\newlength{\queryLength}
\algnewcommand\send{\textbf{Send}\xspace}
\algrenewcommand\textproc{\textsf}
\algrenewcommand\algorithmicwhile{\textbf{while}}
\algrenewcommand\algorithmicfor{\textbf{for}}
\algrenewcommand\algorithmicforall{\textbf{forall}}
\algrenewcommand\algorithmicfunction{\textbf{Function}}
\algrenewcommand\algorithmicprocedure{\textbf{Procedure}}
\algrenewcommand\algorithmicreturn{\textbf{return}}
\algrenewcommand\algorithmicelse{\textbf{else}}
\algrenewcommand\algorithmicif{\textbf{if}}
\algnewcommand\algorithmicforeach{\textbf{for each}}
\algdef{S}[FOR]{Foreach}[1]{\algorithmicforeach\ #1\ \algorithmicdo}
\usepackage[
    lambda, 
    advantage, 
    operators, 
    sets,
    adversary, 
    notions, 
    primitives, 
    asymptotics,
    landau,
    complexity]
{cryptocode}
\renewcommand{\secpar}{\ensuremath{\kappa}}
\newcommand{\db}{\ensuremath{\mathcal{D}}\xspace}
\newcommand{\sk}{\ensuremath{\mathsf{sk}}\xspace}
\newcommand{\domain}{\ensuremath{\mathcal{X}}\xspace}
\newcommand{\client}{\ensuremath{\mathcal{C}}\xspace}
\newcommand{\server}{\ensuremath{\mathcal{S}}\xspace}
\newcommand{\ds}{\ensuremath{\mathcal{DS}}\xspace}
\newcommand{\op}{\ensuremath{\mathsf{op}}\xspace}
\newcommand{\ops}{\ensuremath{\mathsf{ops}}\xspace}
\newcommand{\upt}{\ensuremath{\mathsf{updt}}\xspace}
\newcommand{\isrt}{\ensuremath{\mathsf{insert}}\xspace}
\newcommand{\del}{\ensuremath{\mathsf{delete}}\xspace}
\newcommand{\upd}{\ensuremath{\mathsf{update}}\xspace}
\newcommand{\setup}{\ensuremath{\Pi_{\mathsf{setup}}}\xspace}
\newcommand{\query}{\ensuremath{\Pi_{\mathsf{query}}}\xspace}
\newcommand{\update}{\ensuremath{\Pi_{\mathsf{update}}}\xspace}
\newcommand{\access}{\ensuremath{\mathcal{AP}}\xspace}
\newcommand{\hashmap}{\ensuremath{\mathsf{map}}\xspace}
\newcommand{\tags}{\ensuremath{\mathsf{tags}}\xspace}

\newcommand{\pancake}{\textsc{Pancake}\xspace}

\newcommand{\ie}{\textit{i}.\textit{e}.\xspace}
\newcommand{\eg}{\textit{e}.\textit{g}.\xspace}

\newcommand{\system}{\textsc{Swat}\xspace}

\newcommand{\para}[1]{{\bf \noindent #1 }} 
\newcommand{\epsolute}{\ensuremath{\mathcal{E}}psolute\xspace}

\newtheorem{theorem}{Theorem} 
\newtheorem{definition}{Definition}
\pgfplotsset{compat=newest}
\renewcommand{\descriptionlabel}[1]{\hspace{\labelsep}\emph{#1}}

\usepackage{tikz}
\usetikzlibrary{automata,positioning,arrows}
\tikzset{->-/.style={decoration={
            markings,
            mark=at position #1 with {\arrow{>}}},postaction={decorate}}}

\makeatletter
\newcommandx*{\sendmessageboth}[2][1=<->]{%
    \sendmessage{#1}{#2}%
}
\WithSuffix\newcommand\sendmessageboth*[2][\pcdefaultmessagelength]{%
    \begingroup%
        \renewcommand{\@pcsendmessagetop}{\let\halign\@pc@halign$\begin{aligned}#2\end{aligned}$}%
        \sendmessage{<->}{length=#1}%
    \endgroup%
}
\makeatother

\newcommand{\myappend}{\text{Appendix}\xspace}
\newcommand{\revise}[1]{#1}

\begin{document}
\title{\texorpdfstring{\system}{SWAT}: A System-Wide Approach to Tunable Leakage Mitigation in Encrypted Data Stores}

\author{Leqian Zheng}
\orcid{0000-0003-3772-4400}
\affiliation{%
  \institution{City University of Hong Kong}
}
\email{leqian.zheng@my.cityu.edu.hk}

\author{Lei Xu}
\orcid{0000-0001-9178-6640}
\affiliation{%
  \institution{Nanjing University of Science and Technology}
}
\email{xuleicrypto@gmail.com}

\author{Cong Wang}
\orcid{0000-0003-0547-315X}
\affiliation{%
  \institution{City University of Hong Kong}
}
\email{congwang@cityu.edu.hk}

\author{Sheng Wang}
\affiliation{%
  \institution{Alibaba Group}
}
\email{sh.wang@alibaba-inc.com}
\author{Yuke Hu}
\affiliation{%
  \institution{Zhejiang University}
}
\email{yukehu@zju.edu.cn}
\author{Zhan Qin}
\affiliation{%
  \institution{Zhejiang University}
}
\email{qinzhan@zju.edu.cn}

\author{Feifei Li}
\affiliation{%
  \institution{Alibaba Group}
}
\email{lifeifei@alibaba-inc.com}
\author{Kui Ren}
\affiliation{%
  \institution{Zhejiang University}
}
\email{kuiren@zju.edu.cn}

\begin{abstract}
  Numerous studies have underscored the significant privacy risks associated with various leakage patterns in encrypted data stores.
  While many solutions have been proposed to mitigate these leakages, they either (1) incur substantial overheads, (2) focus on specific subsets of leakage patterns, or (3) apply the same security notion across various workloads, thereby impeding the attainment of fine-tuned privacy-efficiency trade-offs.
  In light of various detrimental leakage patterns, this paper starts with an investigation into which specific leakage patterns require our focus in the contexts of key-value, range-query, and dynamic workloads, respectively.
  Subsequently, we introduce new security notions tailored to the specific privacy requirements of these workloads.
  Accordingly, we propose and instantiate \system, an efficient construction that progressively enables these workloads, while provably mitigating system-wide leakage via a suite of algorithms with tunable privacy-efficiency trade-offs.
  We conducted extensive experiments and compiled a detailed result analysis, showing the efficiency of our solution.
  \system is about an order of magnitude slower than an encryption-only data store that reveals various leakage patterns and is two orders of magnitude faster than a trivial zero-leakage solution.
  Meanwhile, the performance of \system remains highly competitive compared to other designs that mitigate specific types of leakage.

\end{abstract}

\maketitle

\section{Introduction}\label{sec:introduction}

With the advent of cloud computing, many companies and institutions are outsourcing databases and workloads from their private data centers to the cloud. While this transition offers advantages such as increased availability, scalability, and cost-effectiveness, it also exposes users to potential privacy breaches and data abuse. 
Consequently, there has been significant progress %
in constructing \emph{encrypted databases} to prevent adversaries with high privileges or even physical access to the server from learning sensitive data. 
Specifically, one line of work~\citep{Popa2011CryptDBPC, Poddar2016ArxAS} leverages specialized cryptographic primitives to perform different operations over ciphertexts. 
Another prosperous line of work~\citep{Bajaj2014TrustedDBAT, Zheng2017OpaqueAO, Priebe2018EnclaveDBAS, Gribov2019StealthDBAS, Eskandarian2019ObliDBOQ, Antonopoulos2020AzureSD, Sun2021BuildingES, Wang2022OperonAE} we will follow utilizes trusted execution environments (TEE), \eg, Intel SGX and AMD SEV, to process confidential data as plaintext in an isolated and protected approach.

Unfortunately, despite powerful enclaves, recent studies show that encrypted databases still exhibit diverse leakage patterns, including %
1) memory access pattern indicates which memory blocks are accessed, potentially revealing sensitive information like the access frequency of encrypted records;
2) volume pattern refers to the size of the query result set, which can be easily obtained by observing network communication;
3) order pattern refers to the ordinal relationship between data, which may be inferred from data storage (\eg, encrypted albeit sequentially stored data) or memory accesses (\eg, a search over a B+ tree directly reveals the exact ordinal relationship between nodes along the path); 
4) query correlation pattern reveals how queries are correlated, \eg, humans or workloads often generate queries based on previous ones;
and 5) operation timestamp pattern denotes when the encrypted database is accessed. These leakage patterns pose risks for adversaries to recover confidential queries or data through leakage attacks~\citep{Islam2012AccessPD, Cash2016LeakageAbuseAA, Grubbs2017LeakageAbuseAA, Roessink2020ExperimentlRI, Grubbs2018PumpUT, Gui2019EncryptedDN, Blackstone2020RevisitingLA, KornaropoulosPT21, GuiPT23, Lacharit2018ImprovedRA, Grubbs2019LearningTR, KornaropoulosPT19, Oya2021IHOPIS, Kamara2023Maple, Grubbs2016BreakingWA, Wang2021DPSync, xu2023leakage}.

Given its critical importance, many solutions~\citep{Zheng2017OpaqueAO, Mishra2018OblixAE, Kamara2019ComputationallyVS, Eskandarian2019ObliDBOQ, Patel2019MitigatingLI, Grubbs2020PancakeFS, ZhaoWL21, AmjadPPYY21, Qin2022Adore} have thus been proposed to thwart these attacks by eliminating or mitigating these leakage patterns.
From the perspective of performance, attaining full leakage suppression for even a single pattern entails substantial overhead, some of which are even insurmountable. 
For instance, the well-known logarithmic lower bounds for ORAM~\citep{Goldreich1987TowardsAT, Larsen18YesORAMLowerBound}) almost make it theoretically infeasible to provide obliviousness in large-scale databases. 
And full query decorrelation in a known query distribution has been shown in~\cite{Grubbs2020PancakeFS} to be as hard as the offline ORAM.
Fortunately, it is unnecessary to offer full leakage suppression in many scenarios since (1) the adversary's auxiliary information about the encrypted data and/or query workloads is practically biased or distorted, and (2) an adversary has to accumulate sufficient leakage to launch risky leakage attacks. 
For instance, some known-data attacks~\cite{Islam2012AccessPD, Blackstone2020RevisitingLA, Roessink2020ExperimentlRI} assume explicit knowledge of a (probably large) subset of the encrypted data or queries, which seems too strong in reality.
\citet{Kellaris2016GenericAO} show that an adversary needs $\Omega(n^4)$ range queries (exposing access patterns) to reconstruct the exact value of every record in an encrypted database of size $n$. 
And recovering data values with a relative error $\epsilon$ needs $\Omega(\epsilon^{-4})$ queries~\citep{Grubbs2019LearningTR}.
Hence, mitigating partial yet significant leakage with manageable performance overheads is destined and admissible. 
Considering diverse and complex deployment scenarios, such a trade-off between performance and security should ideally be tunable. 

From the perspective of system leakage, most countermeasures protect only subsets of leakage patterns. 
For instance, ObliDB~\citep{Eskandarian2019ObliDBOQ} provides a set of customized oblivious operators to conceal memory access patterns across various query workloads, yet it ignores hiding the sizes of the intermediate and result tables (\ie, volume pattern). 
Most volume-hiding solutions~\citep{Kamara2019ComputationallyVS, Patel2019MitigatingLI, AmjadPPYY21, ZhaoWL21} ignore the query equality pattern, which indicates whether two queries repeat. 
\pancake~\cite{Grubbs2020PancakeFS} mitigates the access pattern via smoothing the access frequency to entries in an encrypted data store, while the exposed query correlation pattern has been shown to be vulnerable to IHOP attack~\cite{Oya2021IHOPIS}. 
It is hence essential to consider system-wide leakage while designing encrypted data stores.
This problem is more pronounced in leakage mitigation schemes, where relaxing protection over existing leakage may inadvertently expose new and detrimental leakage patterns.
A notable example is that relaxing oblivious access to frequency-smoothed access introduces the query correlation pattern~\cite{Grubbs2020PancakeFS, Oya2021IHOPIS}. 

From the perspective of protection applied, existing solutions typically apply a uniform security notion over all supported workloads. 
This approach simplifies the comprehension of the system's security but impedes the attainment of fine-tuned and improved privacy-efficiency trade-offs. 
For instance, Adore~\cite{Qin2022Adore} protects a set of common workloads in relational databases in a differentially oblivious approach, which ensures that the access pattern complies with the differential privacy notion. 
However, some workloads, such as table joins that do not preserve neighboring outputs (as discussed in~\cref{sec:relatedWork}), may not align well with such a protection measure.
Besides, \epsolute~\cite{Bogatov2021Epsolute} applies the same differentially private sanitizer to both point and range queries, while we could leverage a more efficient scheme~\cite{Patel2019MitigatingLI} to hide the volume pattern in point queries. 
It is hence valuable to identify the nuanced privacy requirements inherent in each specific workload, with the primary challenge being that efficiently and securely accommodating a new workload 
may necessitate modifications to the existing ones.

\subsection{Our Contributions}\label{subsec:contribution}
Drawing on the above insights, we investigate the system-wide approach towards tunable leakage mitigation by following recent enclave-based encrypted data stores~\citep{Eskandarian2019ObliDBOQ, Gribov2019StealthDBAS, Antonopoulos2020AzureSD, Wang2022OperonAE}. 
To this end, we present \system, an efficient encrypted data store that \revise{\emph{progressively}} supports key-value, range-query, and dynamic workloads with tunable system-wide leakage mitigation.
We adopt a widely accepted assumption~\cite{Grubbs2020PancakeFS, Bogatov2021Epsolute} that posits the existence of a \emph{trusted client proxy}.
The proxy is responsible for routing client queries to the enclave deployed on the cloud server via a secure and authenticated channel.
We additionally assume a \emph{dedicated communication channel} as the billing model based on network bandwidth (rather than data transferred) usage per month (or year), which we denote as pay-by-bandwidth, is a standard practice in cloud services\footnote{Examples in alphabetical order include \href{https://www.alibabacloud.com/help/en/elastic-compute-service/latest/public-bandwidth}{Alibaba Cloud}, \href{https://aws.amazon.com/directconnect/pricing/?nc=sn&loc=3}{AWS Direct Connect}, \href{https://azure.microsoft.com/en-us/pricing/details/expressroute/}{Azure ExpressRoute}, and \href{https://www.tencentcloud.com/document/product/213/10578}{Tencent Cloud}.}.

Our work starts from the key-value workload due to its simplicity, where keys and equal-length values are protected by pseudorandom functions and authenticated encryption. 
We mitigate the primary leakage, \ie, access and query correlation patterns, via frequency smoothing and partial query decorrelation respectively. 
We adopt \pancake for the former one due to its noteworthy gains in balancing performance and security. 
\pancake provides provable assurance of achieving a uniform access frequency for each entry in the key-value store, irrespective of their original access distribution, as long as each query is generated independently.
However, it falls short in protecting encrypted data stores when queries are correlated~\cite{Oya2021IHOPIS}.
We hence devise an elegant and almost-for-free security patch, modeled as $\theta$-query decorrelation, to mitigate this leakage pattern. 
This technique also holds potential for broader applications in other searchable encryption systems.

We then extend \system to handle range queries that inherently expose more leakage, such as order and volume patterns. 
To system widely address these leakage patterns, we introduce a formal security notion aimed at reducing leakage of range query systems to that of the previous stage (\ie, key-value stores). 
Intuitively, no adversary under this notion can distinguish between a sequence of data accesses from range queries (sampled from an arbitrary distribution) and those from uniformly random sampling. 
To achieve this, we develop an efficient protocol that sets up the encrypted data store by partitioning the input dataset into buckets without revealing their order, and handles queries by accessing the data store at a fixed rate and retrieving a fixed number of buckets in each access.
This protocol effectively suppresses order, volume, and search timestamp patterns, with reasonable monetary cost (thanks to the pay-by-bandwidth billing model). 

Subsequently, we introduce a data-structure dynamization technique to enable updates over the encrypted data store.
In order to maintain query efficiency, the data store requires a well-structured, albeit hidden from anyone but the trusted proxy (for security), search index. 
Unless properly mitigated, the memory access pattern during index updates will expose sensitive information regarding the underlying structure of the encrypted data.
We capture such a dominant privacy demand by formulating a differential obliviousness notion~\cite{ChanCMS19} since it offers principled privacy-efficiency trade-offs. 
We also evolve a $k$-way differentially oblivious merge algorithm from the $2$-way one~\citep{ChanCMS19}, which serves as the foundation of dynamization, to offer better privacy guarantees. 
Moreover, we improve the practical performance of the differentially oblivious merge algorithm by notably reducing the size of its oblivious buffer without hurting the security guarantees it claims. 

We then implement an end-to-end system \system that realizes the above functionalities on top of Intel SGX. 
\system is about $10.6\times$ slower than an encryption-only database that exposes all detrimental leakage patterns and $31.6\times$ faster than a trivial solution that eliminates all these patterns.
We also compare it to ObliDB~\cite{Mishra2018OblixAE}, the state-of-the-art oblivious database, and \epsolute~\cite{Bogatov2021Epsolute}, the state-of-the-art range-query system mitigating the volume pattern and eliminating the memory access pattern.
The result shows that our design provides competitive performance while mitigating system-wide leakage patterns.  
We also run extensive experiments with various settings and compiled a detailed result analysis. 

We summarize our contributions in this work as follows:
\begin{itemize}[topsep=0pt, partopsep=0pt, leftmargin=10pt]
    \item We customize a set of security models to capture varying privacy requirements in key-value, range-query, and dynamic workloads.
    \item We present \system, an efficient design that \emph{progressively} enables these workloads,
    \revise{offering tunable privacy-efficiency trade-offs and strategies for their systematic organization and integration.} 
    \item We implement \system and empirically evaluate its performance over an extensive set of settings with a detailed results compilation showing the efficiency of our construction and the tunable privacy-efficiency trade-offs.
\end{itemize}

\section{Background}\label{sec:preliminaries} 
In this section, we first describe an outsourced data store system adapted from~\cite{Bogatov2021Epsolute}.
Then we introduce the threat model and formal security notions that capture different privacy requirements across various workloads. 
\subsection{Syntax and System Model}
Without the loss of generality, we abstract a data store as a collection of $n$ records $r$, each with a search key $\sk$: $\db = \allowbreak\{(\sk_1,\allowbreak r_1),\allowbreak \dots,\allowbreak (\sk_n,\allowbreak r_n)\}$. 
We assume that search keys take value from a \emph{well ordered} domain $\domain=\set{1, \dots, N}$ for $N\in\mathbb{N}$, and all records have the same fixed bit-length. 
We describe the model for a single indexed attribute for ease of presentation and discuss how it can be extended to support multiple attributes. 

We explicitly distinguish operations $\op$ over the data store as queries $q$ and updates $u$. 
A query is a predicate $q:\domain\to\set{0, 1}$ to be evaluated on $\db$. 
It results in a set $q(\db)=\set{r_i: q(\sk_i)=1}$ containing all records whose search keys are evaluated to be true.
This work focuses on the following types of queries (adapted from~\cite{Bogatov2021Epsolute}).
\noindent\emph{Range query.}
A range query $q_{[x, y]}(a)$ associated with an interval $[x, y]$ is evaluated to $1$ iff $x\leq a\leq y$.
The equivalent SQL query is:

\texttt{\small
SELECT * FROM tab WHERE attr BETWEEN x AND y.
}

\noindent\emph{Point query.}
A point query $q_{x}(a)$ associated with an element $x\in\domain$ is evaluated to $1$ iff $x=a$.
The equivalent SQL query is:

\texttt{\small
SELECT * FROM tab WHERE attr = x.
}

An update operation, denoted as $u=(\upt, \sk, r),\allowbreak \upt\in\{\isrt,\allowbreak \del,\allowbreak \upd\}$, performs one of three actions on the data store $\db$ that results in an updated data store $\db'$. 
According to the update type, we have:

\noindent\emph{Insertion.}
An insertion $u=(\isrt, \sk, r)$ results in $\db'=\db \cup \set{(\sk, r)}$. 
The equivalent SQL statement is: 

\texttt{\small
INSERT INTO tab VALUES ($\sk$, $r$).
}

\noindent\emph{Deletion.}
A deletion $u=(\del, \sk, \bot)$ results in a (probably) new database $\db'\subseteq\db$ such that for all $(\sk_i, r_i)\in\db\setminus\db'$, we have $\sk_i=\sk$. 
The equivalent SQL statement is: 

\texttt{\small
DELETE FROM tab WHERE attr = \sk.
}

\noindent\emph{Update.}
An update $u=(\upd, \sk, r')$ results in a new database $\db'$ such that for all $(\sk_i, r_i)\in \db$, we have $(\sk_i, r_i)\in\db'$ if $\sk_i\neq \sk$ and $(\sk_i,r')\in\db$ if $\sk_i=\sk$. 
The equivalent SQL statement is:

\texttt{\small
UPDATE tab SET r = $r'$ WHERE attr = \sk.
}

\para{Outsourced dynamic data store (ODDS).} An ODDS consists of three protocols between two \emph{stateful} parties: a client \client and a server \server (adapted from~\cite{Bogatov2021Epsolute}).

\noindent\emph{Setup protocol} \setup: \client takes as input a database $\db$ (and parameters for other purposes); \server takes no input. 
\client has no output (except its state); \server outputs a data structure \ds.

\noindent\emph{Query protocol} \query: \client has a query $q$; \server has as input \ds. 
\client outputs $q(\ds)$; \server has no output. 
Both may update internal states.

\noindent\emph{Update protocol} \update: \client has an update $u$; \server has input \ds. 
\client has no formal output; \server outputs an updated data structure $\ds'$. 
Both may update their internal states.

\para{Correctness.}We require that for any database and any operation sequence consisting of queries and updates, it holds that running \setup, and then \query and \update on the corresponding inputs, \query outputs the correct results except with negligible probability over the coins of the above runs. 

\para{Efficiency.}We measure the efficiency of an ODDS from the following perspectives: 
1) \emph{Storage efficiency} measures the bit lengths of an ODDS in \client and \server, including their states.
The storage complexity of \client should be significantly smaller than the bit length of the data store $|\db|$.
2) \emph{Communication efficiency} measures the necessary network bandwidth of an ODDS rather than the bit lengths of data transferred between \client and \server.
We prefer the bandwidth metric as pay-by-bandwidth is a common billing model in cloud services (as discussed in~\cref{subsec:contribution});
3) \emph{Search time efficiency} measures the time span between when a client issues a query and when it receives the corresponding results.
4) \emph{Update time efficiency} measures an insertion, deletion, or updation's (amortized) processing time in \server. 

\subsection{Threat Model and Security Definitions}\label{subsec:securityDef}

\para{Threat model.}We use Intel SGX as an example of hardware enclaves to discuss our threat model. 
Intel SGX offers confidentiality and integrity of data and codes inside its protected memory, \ie, enclave page caches (EPC).
An enclave is defined by user-level or operating system code, initiated by loading a verifiable compiled library, and interacted via well-defined functions.
EPC is ``uni-directly'' accessible, \ie, codes inside EPC can access the entire address space (except those belonging to other enclaves), but the others cannot access EPC. 
SGX enables a remote system to verify what code is loaded into an enclave and set up a secure communication channel with the enclave via remote attestation. 
Furthermore, 
the capacity of EPC is highly limited (\ie, 128MB in total and less than 100 MB available) compared to the untrusted memory. 
SGX v2 offers much richer EPC resources (up to 512 GB per processor) by dropping the integrity guarantee inside the enclave. 
We also take limited EPC into account while \revise{implementing} \system.

Similar to prior works \citep{Mishra2018OblixAE, Eskandarian2019ObliDBOQ, Antonopoulos2020AzureSD, Sun2021BuildingES}, we assume an \emph{honest-but-curious adversary} with the power to continuously inspect network communication, untrusted memory, and disk, and data transferred inside the system bus (a.k.a., persistent adversary). 
In particular, both data flowing through the data bus and memory addresses carried in the address bus could be observed by adversaries. 
The latter one exposes the memory access patterns to both trusted and untrusted memory~\cite{Bulck2017TellingYS, Xu2015ControlledChannelAD, Mishra2018OblixAE}. 
Adversaries can also leverage arbitrary auxiliary information (\eg, a subset of encrypted data or queries, or a probably biased distribution of data records or client queries) to recover data or queries.

Timing side channels and power analysis are orthogonal to our work similar to previous studies~\cite{Mishra2018OblixAE, Eskandarian2019ObliDBOQ, Antonopoulos2020AzureSD, Sun2021BuildingES}. 
Denial of service attacks, which can be easily launched by a privileged attacker against an enclave, are out of scope since it does not compromise user privacy. 
Although there are several side-channel attacks against SGX upon speculative execution, branching history, or page faults~\cite{brasser2017software, lee2017inferring, Bulck2017TellingYS, Weichbrodt2016AsyncShockES, Xu2015ControlledChannelAD}, effective approaches \cite{Rane2015RaccoonCD,  Shinde2016PreventingPF, Seo2017SGXShieldEA, Shih2017TSGXEC} have been proposed to mitigate these attacks and we can employ a more secure SGX implementation if necessary. 

\para{Frequency smoothing.}To hide the access \emph{distribution} of items in an encrypted key-value store, \citet{Grubbs2020PancakeFS} proposed a security notion named ``real-or-random indistinguishability under chosen (dynamic) distribution attack''. 
Informally, it requires that no adversary is able to distinguish whether a sequence of accessed items are queried by clients or randomly sampled from a uniform distribution, \ie, any characteristic distribution over encrypted items that contains sensitive information will be smoothed to a uniform one. 

\pancake achieves this goal by selective replication upon initialization and batched query strategy with fake queries. 
Specifically, selective replication creates copies of items that are more likely to be accessed, and accesses one of them if queried. 
High likelihood is hence amortized to the average. 
It then creates fake queries for items that are less likely to be accessed, which hence raises the low likelihood to the average. 
Batches are introduced to ensure that fake queries are indistinguishable from real ones, and that real queries will be answered timely.
The storage overhead caused by replicas could be bounded by a constant. 
One may also trade storage overhead $\alpha$, which is defined as the total number of replicas divided by the original count, for reduced communication overhead via fewer fake queries (or vice versa). 
The authors use $\alpha=2$ as default in their design and experiments. 

\para{Query decorrelation.}As discussed earlier, \pancake is vulnerable when queries are correlated. 
Rather than a full query decorrelation notion that is as hard as offline ORAM~\citep{Grubbs2020PancakeFS}, we propose the following partial query decorrelation notion, wherein we require the current query to exhibit independence from a minimum of $\theta$ previous queries.
It trivially captures the decorrelation requirement since independence implies zero correlation. 
\begin{definition}[$\theta$-query decorrelation]\label{def:decorrelation}For a discrete-time stochastic query process $\{X_t\in \mathcal{Q}:t\geq 0\}$ where $\mathcal{Q}$ denotes the countable set of possible queries, we say that it satisfies $\theta$-query decorrelation if for all $t\in\mathbb{N}_+$ and $q_{0}, \dots, q_{t-1}\in\mathcal{Q}$, there exists $S\subseteq [t]$ with $t^{\prime}=|S|$ and $t^{\prime}\leq \max(t-\theta, 0)$ such that $\Pr\bigl[X_{t}=q_{t}|X_0=q_{0}, \dots, X_{t-1}=q_{t-1}\bigr]=\Pr\bigl[X_{t}=q_{t}|X_{s_0}=q_{s_0}, \dots, X_{s_{t^{\prime}-1}}=q_{s_{t^{\prime}-1}}\bigr].$
\end{definition}
\para{Range or random point query indistinguishability.}
We propose a formal security model that captures the indistinguishability between data accesses from range queries (generated from a specific distribution) and data accesses from uniformly random sampling (detailed in \myappend).
Achieving this security goal rules out attacks based on order or volume patterns, as such leakage patterns do not even exist in individual data access systems (\ie, the key-value stores) with values of the same length.

\para{Differential obliviousness (DO).}DO~\cite{Kellaris2016GenericAO, WaghCM16, wagh2018differentially, ChanCMS19, PersianoY19, Persiano2022LowerBF} essentially requires that the memory access pattern of an algorithm or data structure complies with the well-known differential privacy \cite{Dwork2006Calibrating, Dwork2014Algorithmic} notion, which protects the privacy of individuals in published results by adding noise to the data. 
Since the cloud service provider is untrusted, we require the access patterns observed during the data store operation rather than the outputs of computation, to satisfy differential privacy. 

Two operational sequences $\ops$ and $\ops'$ consisting of the same number of queries are called neighboring if they differ in exactly one position $i$, and both are of the same update query type (insertion or deletion). 
Indeed, such two neighboring sequences denoted by $\ops\sim\ops'$ over the same setup dataset will result in two neighboring data stores that differ in exactly one record. 
\begin{figure*}
    \centering
    \includegraphics[width=0.9\textwidth]{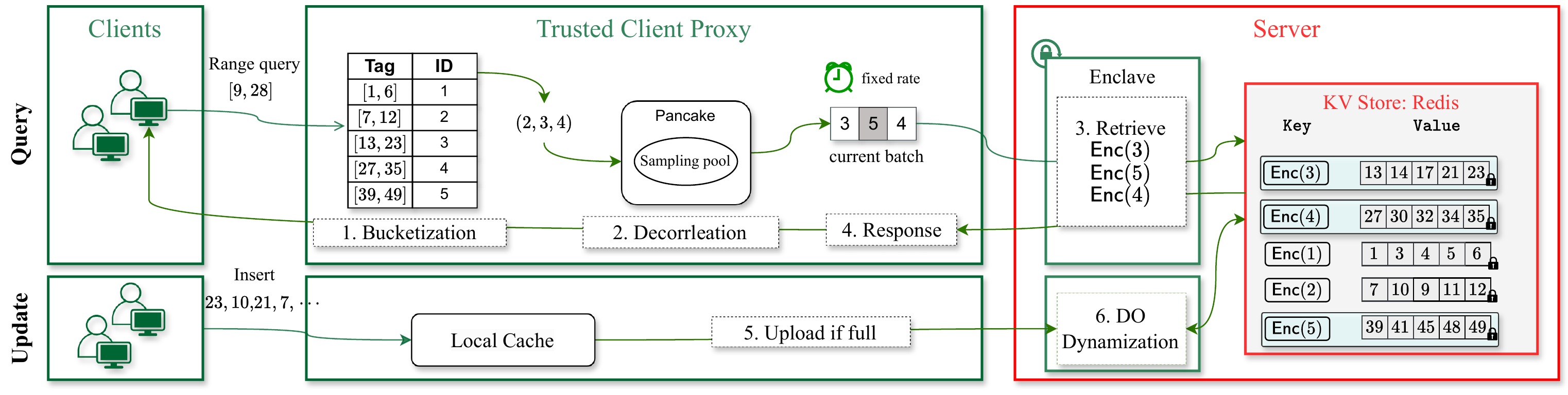}
    \caption{
    Upon receiving a query, the client proxy splits it into bucket accesses based on tags prepared by the cloud server's enclave. These accesses are added to a sampling pool, from which a batch of bucket IDs is randomly selected. The proxy also filters false positives in those buckets and returns the final results to the client.
    Upon receiving an update, the client proxy caches it until a bucket of updates accumulates. Subsequently, the client proxy sorts these updates and uploads them to the server's enclave. The enclave then updates the encrypted data store in a differentially oblivious way.
    }
    \label{fig:baseComponent}
\end{figure*} 

\begin{definition}\label{def:DODS} 
We say that a dynamic outsourced data store $\Pi$ is $(\varepsilon, \delta)$-differentially oblivious with respect to updates (a.k.a DO$_{\text{update}}$-ODDS) if for any data store $\db$ and any two query-consistent neighboring operational sequences $\ops\sim\ops'$, and any possible set of memory access patterns $S$ (adapted from~\cite{ChanCMS19}): 
\begin{equation*}
  \Pr[\access_{\Pi}(\db, \ops)\in S]\leq \text{exp}({\varepsilon})\cdot \Pr[\access_{\Pi}(\db, \ops')\in S] + \delta.
\end{equation*}
\end{definition}
The parameter $\varepsilon$ as a privacy loss metric adjusts the efficiency-privacy trade-off. 
The parameter $\delta$ allows for a negligible probability when the bound $\varepsilon$ fails to hold, and allowing such a negligible failure probability is essential to improve system performance, as shown in~\cite{ChanCMS19}. 
The random variable $\access_{\Pi}(\db, \ops)$ denotes the distribution of access patterns incurred by the system $\Pi$ over $\db$ and $\ops$.

\section{\texorpdfstring{\system}{Swat} Design}\label{sec:range2kv}

\system is constructed in a progressive approach that enables key-value~\cref{subsec:queryDecorrelation}, range-query~\cref{subsec:rangeQuerySupport}, and dynamic~\cref{subsec:dynamization} workloads with emphasis on efficiently mitigating system-wide leakage.
Before delving into concrete construction, we first provide an overview of \system, the one supporting all the above workloads, to help better understand how it works. 

\subsection{Overview} 
We provide an illustration of \system's workflow in~\cref{fig:baseComponent}. 
As introduced above, \system assumes a \emph{trusted} proxy that routes queries from clients to the untrusted cloud server. 

\para{Workflow.}\system sets up an encrypted data store over a sorted input dataset by dividing it into buckets of the same size, randomly shuffling those buckets to hide the order pattern. 
Similar to \pancake~\cite{Grubbs2020PancakeFS}, each bucket will be encrypted by an authenticated encryption scheme $E$ and inserted into the backend KV store with an identifier generated by a secretly keyed pseudorandom function $F$. 
Upon receiving a range query from a client, the client proxy will partition it into several bucket requests according to their tags (step 1).  
It then puts bucket requests into a sampling pool for future processes (step 2). 
At each predefined time unit, it samples a batch of pending buckets as well as fake buckets and retrieves them from the backend KV store on the cloud server via a secure enclave (step 3). 
Once a range query receives all requested buckets, correct answers will be replied to the corresponding client by filtering the false positives (step 4). 
Updates over the encrypted data store are more complicated. 
The client proxy will maintain a small local buffer to cache newly inserted records, which will be uploaded to the cloud server as it fills up (step 5). 
The secure enclave in the untrusted server will then fetches necessary encrypted data, decrypt it inside its private memory, and merge several sorted record arrays into a new one in a differentially oblivious manner (step 6). 
Finally, the client proxy will update the necessary local states accordingly. 

\para{Parameters.}We also brief system parameters to promote a high-level understanding of how they tune the performance and security of \system.
A larger $\theta$ (\cref{subsec:queryDecorrelation}), which denotes the minimum number of items pending in the sampling pool, implies a stronger decorrelation effect that comes at the cost of reduced search time efficiency.
A larger bucket size $Z$ (\cref{subsec:rangeQuerySupport}) signifies fewer states cached in the client proxy but more false positives retrieved from the server. 
A large privacy loss $\varepsilon$ as well as a smaller security parameter $\secpar$ (\cref{subsec:dynamization}) allows better update performance at the cost of less obliviousness. 
In addition, a larger $k$ (\cref{subsec:dynamization}), which denotes the number of components for dynamization, improves the update time efficiency but results in a worse search time efficiency and a larger privacy loss. 
Ideally, one should select a suit of parameters and conduct a benchmark test in the intended deployment environment to identify the most effective parameters, which may vary due to factors such as storage access times and network round trip times (RTTs). 

\subsection{$\theta$-query Decorrelation in Frequency-Smoothed KV Stores}\label{subsec:queryDecorrelation}
As outlined in~\cref{subsec:securityDef}, the primary leakage sources in a key-value store with a uniform value length are access frequency and query correlation patterns. 
The former could be efficiently fixed via the frequency smoothing technique~\cite{Grubbs2020PancakeFS}. 
However, this technique offers limited protection against query correlation, primarily due to its reliance on a \emph{queue} to cache pending items for future batches. 
Specifically, once a query arrives, the client proxy will randomly choose a replica of the key, add it to the pending queue, and prepare a batch of accesses to the cloud server. 
For each access in the batch, the proxy randomly flips a coin to determine whether it is fake or real, where fake accesses are generated according to a pre-settled distribution. 
The proxy retrieves elements from the queue as real accesses if the queue is not empty. Otherwise, it will simulate a real query by sampling a key (replica) based on its real distribution. 

\setlength{\textfloatsep}{0pt}
\begin{algorithm}[t]
    \centering
    \caption{Sampling Pool with $\theta$-Decorrelation in \client}\label{alg:samplingPool}
    \small
    \begin{minipage}[t]{0.49\linewidth}
        \begin{algorithmic}[1] 
            \Statex\underline{$\mathsf{Setup}(\theta, \hat{\pi}, \mathsf{fn})$:}
            \Statex \hspace{0em}\Comment{\footnotesize$\hat{\pi}$: the real distribution of replicas} 
            \Statex \hspace{0em}\Comment{\footnotesize$\mathsf{fn}$: the update policy for weights} 
            \State $\hashmap\gets\emptyset$ \Comment{\footnotesize hash table}
            \State $\mathsf{items}\gets[]$ \Comment{\footnotesize item array} 
            \State $\mathsf{w}\gets[]$ \Comment{\footnotesize sampling weights}
            \State $i\gets 0$ 
            \While{$i<\theta$}
            \State $r_{\mathsf{rep}}\gets \hat{\pi}$
            \If{$\mathsf{Put}(r_{\mathsf{rep}})$}
            \State $i\gets i+1$
            \EndIf 
            \EndWhile
        \end{algorithmic} 
    \end{minipage}
    \begin{minipage}[t]{0.49\linewidth}
        \begin{algorithmic}[1]
            \Statex \underline{$\mathsf{Put}(r_{\mathsf{rep}})$:}
            \If{$\hashmap.\mathsf{has}(r_{\mathsf{rep}})$}
            \State \Return false
            \EndIf{}
            \State $\hashmap.\mathsf{insert}(r_{\mathsf{rep}})$ 
            \State $\mathsf{items}.\mathsf{add}(r_{\mathsf{rep}}), \mathsf{w}.\mathsf{add}(1)$ 
            \State \Return true 
        \end{algorithmic}
        
        \begin{algorithmic}[1] 
            \Statex\underline{$\mathsf{Get()}$:}
            \State $i\gets \mathsf{sampleIdx}(\mathsf{items}, \mathsf{w})$
            \State $\hashmap.\mathsf{remove}(\mathsf{items}[i])$
            \State $\mathsf{items}.\mathsf{del}(i),  \mathsf{w}.\mathsf{del}(i)$
            \State $\mathsf{w}\gets\mathsf{fn}(\mathsf{w})$\Comment{\small update weights} 
            \State $i\gets |\mathsf{items}|$
            \State Execute line $5\sim 8$ of $\mathsf{Setup}$
        \end{algorithmic} 
    \end{minipage} 
\end{algorithm}

The key insight here is that the \emph{first-in-first-out} feature of the pending queue completely preserves the original query correlation. 
To alleviate this issue, 
\revise{%
\system{} introduces a modification in \pancake{} by replacing the queue with a sampling pool. 
This adaptation allows for the random selection of pending items in the sampling pool \emph{independently}. 
Consequently, correlations among \emph{queries pending in the pool} are eliminated due to the inherent property of independence that implies zero correlation.
}
We further introduce the following two new features to provide greater flexibility in tuning efficiency and privacy:
    1) Minimum size $\theta$ that we impose on the sampling pool is an important parameter to adjust the trade-off. 
    To meet the minimum size requirement, \system pads the sampling pool with queries yielded by the real distribution (\ie, simulating client behavior). 
    Increasing $\theta$ can reduce query correlation, albeit with a potentially higher search time efficiency.
    2) Customized sampling weight update policy allows \client to dynamically adjust sampling weights over time in diverse manners. 
    For example, weight increasing over time allows a quasi-FIFO effect. 
    Namely, items arriving earlier (first-in) will have a higher weight and are, therefore, more likely to be sampled in the current batch (first-out). 
    It yields improved search time efficiency, provided that intrinsic query correlation remains within an acceptable security range.

\para{Write support and linearizability.}
\pancake~\cite{Grubbs2020PancakeFS}, as the base of our design, supports writes to keys in the KV store via a standard read-and-write technique in the ORAM literature~\cite{Shi2020PathOH}, \ie, each access consists of a read followed by a write\revise{, collectively forming a transaction}. 
Specifically, since only sampled replicas will be written in each batch access, the proxy will maintain an \textsf{UpdateCache} to track which replicas of a key need to be updated in the future in the form of $k\rightarrow (v, \texttt{UpdateMap})$. 
\texttt{UpdateMap} denotes which replicas of $k$ have been updated or need to be updated in the future. 
In each batch access, \pancake will consult \textsf{UpdateCache} to ensure that updated values propagate.
Upon updating all replicas of $k$ to \server, its corresponding entries in \textsf{UpdateCache} will then be removed.

\revise{
In \system, as accesses are randomly sampled from the pool, it is crucial to establish linearizability~\cite{Herlihy1990linearization}, \ie, a database consistency guarantee ensuring that each operation appears to occur atomically and in accordance with the \emph{real-time} ordering.  
We must ensure that a \textsf{Get}($k$) on a key $k$, either before or after a write operation \textsf{Put}($k$) on that key, accurately reflects the value $v$.
Considering that \system samples \emph{key replicas} rather than the underlying read/write operations,  we assign to each pending key $k$ a list $\mathsf{lst}_k$ that tracks the \emph{yet-to-be responded} operations to $k$ in \emph{real-time ordering}, exemplified by $k\rightarrow (\textsf{read}_1, \textsf{read}_2, \textsf{write}_1(v_1), \textsf{read}_3, \textsf{write}_2(v_2))$.

Upon receiving the result of each batch access from \server, the proxy operates on each key $k$ in the batch as follows. 
We denote the value received from \server as $v_0$. 
The proxy first consults \textsf{UpdateCache}[$k$], inherited from \pancake, to verify if there is a cached update to $k$.
If such an update exists, we replace $v_0$ with the cached value.
The proxy then scans the list $\mathsf{lst}_k$ to address the pending operations as follows.
For each read operation in $\mathsf{lst}_k$, the proxy responds with the current value $v_c$. 
When encountering a write operation in $\mathsf{lst}_k$, the proxy updates $v_c$ with the written value and continues.
Upon completing all operations in $\mathsf{lst}_k$, the proxy verifies if $v_c$ equals $v_0$.
If this condition holds, it writes back to \server a re-encrypted $v_0$. 
Otherwise, it overwrites the corresponding value in \textsf{UpdateCache}[$k$] with $v_c$ and writes back to \server an encryption of $v_c$. 
The proxy will also update \textsf{UpdateMap} following the approach used by \pancake.

In the above example with an initially empty \textsf{UpdateCache}[$k$], the proxy first samples a replica of $k$ (rather than an operation on $k$) and obtains its value $v_0$ from \server.  
Then it responds to $\textsf{read}_1$ and $\textsf{read}_2$ with $v_0$, and to $\textsf{read}_3$ with $v_1$. 
After completing all operations in $\mathsf{lst}_k$, the proxy will update \textsf{UpdateCache}[$k$] with $v_2$ and a properly configured $\mathsf{UpdateMap}$.
Subsequently, it will write to a replica in \server with the encryption of $v_2$.
In short, as \system samples \emph{key replicas rather than operations on keys}, and the proxy will respond to operations on keys in real-time ordering as described above, we can claim that \system correctly establishes linearizability. 
}

The design outlined above is formally referred to as $\theta$-decorrelation, as depicted in~\cref{alg:samplingPool}.
The following theorem establishes its security. %
\begin{theorem}\label{thm:decorrelation}
\system with an unweighted sampling policy achieves $\theta$-query decorrelation in~\cref{def:decorrelation}.
\end{theorem}
\begin{proof}
\Cref{thm:decorrelation} holds trivially as all queries pending in the pool, which contains at least $\theta$ queries, are independently and uniformly sampled, and independence implies zero correlation.
\end{proof}
\para{Efficiency.} Note that the runtime complexity of Put operation is $\bigO{1}$ since a hash table supports \revise{constant-time} insertions and removals. 
$\mathsf{Get}$ operation could also be done in \revise{constant} time \revise{when employing} an unweighted sampling policy (\ie, not weight updates) by observing that removing the $i$-th element from a vector could be done by two operations: 1) swapping the $i$-th element and the last one; 2) discarding the last element. 
\revise{However, if weighted sampling with dynamic weights is considered, the time complexity of $\mathsf{Get}$ would be $\bigO{\theta}$ due to the possible linear scan of the weight vector.}
The storage overhead the sampling pool introduces is proportional to the number of pending queries, which is normally negligible compared to the local states for maintaining replica distributions. 

\revise{
We remark that our design is \emph{almost for free} from the following two perspectives. 
In terms of implementation, it requires minimum intrusion on \pancake, as it only requires substitution of the queue with a sampling pool, both of which expose identical interfaces.
Regarding performance, \system with a uniform sampling policy ensures that \textsf{Put} and \textsf{Get} operations impose only a small constant computation overhead on the client proxy. 
We can hence effectively alleviate query correlation leakage with minimum efforts. 
}

\para{Broader interests.}
This design is also compatible with \textsc{Shortstack}~\cite{Vuppalapati2022SHORTSTACKDF}, which is a distributed and fault-tolerant extension of \pancake.
Furthermore, this technology may help mitigate leakage patterns in searchable symmetric encryptions (SSE), where a client searches over the encrypted index on the cloud for documents containing the target keyword, and then retrieves the matching documents accordingly. 
\citet{GuiPPW20} highlighted the leakage patterns in the file retrieval phase. 
Meanwhile, the direct application of leakage-suppression techniques for encrypted indices on file retrievals fails to provide practical performance.
We claim that this issue could be mitigated by adopting our decorrelation technique as follows:
    1) Obtain file IDs by sending (probably fake) search tokens to the cloud (where the \emph{leakage-suppresion} encrypted index takes effect); 
    2) $\mathsf{Put}$ IDs into the sampling pool of an appropriate size $\theta$, where fake IDs are obtained by sending fake tokens rather than directly generating IDs;
    3) $\mathsf{Get}$ (sample) a batch of IDs and send them to the server to retrieve corresponding files. 

The key insight here is that we are able to obfuscate the query boundaries for pending keyword queries, since documents for different keywords may be retrieved in the same batch. 
The co-occurrence pattern, which many attacks~\cite{Islam2012AccessPD, Cash2015LeakageAbuseAA, Pouliot2016TheSN, Blackstone2020RevisitingLA} rely on, will no longer be present for those pending searches. 
Another minor performance benefit is query deduplication, \ie, pending queries targeted for the same item will result in only one data retrieval. 
It would bring more benefits for keyword searches, where querying for ``synonym''s like ``crypto'' and ``encryption'' will result in a notably less total number of file retrievals.

\subsection{Nearly Zero-Leakage Range Query Support}\label{subsec:rangeQuerySupport}
Here we present how to efficiently support nearly zero-leakage range queries based on the above design. 
We emphasize that even an oblivious design without appropriate padding still reveals the volume pattern, which has been demonstrated to be catastrophic under certain adversarial models~\cite{Grubbs2018PumpUT, Gui2019EncryptedDN, KornaropoulosPT21}.
We will go through our construction in terms of leakage patterns exposed in different stages as well as the mitigations we adopt. 

\para{Order leakage in data storage.}
Clearly, it is inevitable to store data in an ordered format for efficient range queries. 
We emphasize here that data should be stored in a way that no one but the trusted part can learn how they are ordered. 
Otherwise, a snapshot adversary could infer the order pattern according to the physical addresses of the data.
Oblivious shuffling (via sorting with random weights) together with a local position map of the data entries could easily fix it within $\bigO{n\log n}\sim \mathcal{O}({n\log^2 n})$ time. 
In specific, the functional goal of oblivious sorting is to take an array $\mathbf{a}$ of $n$ elements as input and output an array $\mathbf{a}'$ that is a permutation $\Pi:\allowbreak [n] \allowbreak\to \allowbreak[n]$ of $\mathbf{a}$, \ie, $\mathbf{a}'_i=\mathbf{a}_{\Pi(i)}$, such that $\mathbf{a}'_1\leq \cdots\leq \mathbf{a}'_n$. 
The obliviousness requires that the distributions of memory access patterns produced by two input arrays of the same length are indistinguishable from each other.

To better boost the performance of our system, we also investigate the actual performance of existing oblivious sorters to select the most efficient one. 
We choose the classic bitonic sort algorithm with $\bigO{n\log^2 n}$ running time and brief the rationals as follows. 
One may choose to skip this part without impeding their comprehension of subsequent sections if the details of the analysis leading to this conclusion are of no interest.
Though sorting networks such as AKS~\cite{Ajtai1983An0L} and Zig-zag sort~\cite{Goodrich2014ZigzagSA} have a better asymptotic time $\bigO{n\log n}$, the huge hidden constants result in extremely poor practical performance. 
Randomized Shellsort~\cite{Goodrich2011RandomizedSA} with a practical time complexity suffers a non-negligible error probability $\bigO{n^{-3}}$ that may incur other security concerns. 
\citet{Asharov2020BucketOS} proposed another bitonic-resembling sorting algorithm requiring $\bigO{n/Z\cdot \log n}$ rounds of interaction, where $Z$ as a security parameter is usually taken to be a few hundred. 
It is hence not the perfect candidate since the interactive process is sensitive to network latency. 
While the non-interactive transformation discussed in~\cite{Asharov2020BucketOS} assumes an inherently oblivious buffer inside the enclave of size $2Z$ that does not fit our threat model. 
Simulating the small oblivious buffer with $O(\log Z)$ multiplicative overhead will bring a significant performance penalty. 
One can also leverage a path oblivious heap~\cite{Shi2020PathOH} to implement an oblivious sorter. 
But it is non-trivial to parallelize a heap-based sorting algorithm while the bitonic sort could be easily parallelized. 
Besides, there does exist a sorting algorithm~\cite{ChanCMS19} relaxing the security to differential obliviousness (DO) with a potentially better asymptotic time complexity. 
However, we claim that so far it may have only theoretical interest.
Specifically, they propose a stable $(\varepsilon, \delta)$-DO compact algorithm with 
\begin{sloppy} 
$\bigO{n \cdot \log( \varepsilon^{-1} \cdot \log^{1.5}n \cdot  \log\delta^{-1})}$
\end{sloppy} time complexity. 
They then extend it to a $k$-bit-key $(\varepsilon, \delta)$-DO sorter following the approach of the Radix sort.
It overcomes the well-known $\bigO{n\cdot\log n}$ comparison-based sorting barrier only if $k=\smallO{\log n/\log\log n}$ and $\varepsilon=\bigO{1}$. 
However, the former short-key condition is not that common in real-world applications. 
In addition, their $k$-bit DO sorter demands a \emph{stable} oblivious sorter to accomplish their goal, which further complicates the whole procedure since none of the aforementioned oblivious sorters is inherently stable.

\para{Frequency and order leakage in data access.}Range queries inherently exhibit characteristic data access frequencies.
For instance, uniformly random range queries over the domain $\set{1,\dots, N}$ will access the value $1\leq x\leq N$ with probability $p(x)=2x (N + 1 - x) / ( N ( N + 1 ) )$. 
Several query recovery attacks~\cite{Lacharit2018ImprovedRA, Grubbs2019LearningTR} rely heavily on this frequency leakage. %

Another important leakage is that data tend to be accessed sequentially without proper protection. 
It directly reveals the exact order among those data records. 
Fortunately, frequency smoothing combined with our almost-for-free query decorrelation technique allows us to mitigate such leakages effectively. 
For a range query $[l, r]$, we simply put all keys $x\in[l,r]$ into the sampling pool and query them in batch. 
Unweighted random sampling will produce a sequence of data accesses equivalent to random shuffling, which hence fully hides the order leakage. 
It will, of course, introduce notable performance overhead since the result might be very sparse compared to the queried range. 
In addition, it cannot be directly applied to range queries over real numbers due to an infinite number of keys in between or innumerable keys for the floating type.
We will address these issues shortly. 

\para{Volume, order, and timestamp leakage in data transition.}Dealing with volume pattern leakage is more challenging in data transition for the following reasons: 1) Full padding implies pulling the entire database or prohibitive overhead if the maximum possible result set size is huge. 
2) Padding in a differentially private way may not provide sufficient protection. 
3) Replying with a subset of results introduces unacceptable false negatives in various scenarios, or replying with a sequence of subsets to avoid false positives merely coarsens the volume pattern to the subset size level. 

Na\"ively asking for multiple batches consecutively does not perfectly fix the issues in the third approach, as an adversary could easily infer that buckets in a short time window indeed compose a range query and are \emph{in order}. 
Meanwhile, another important leakage we must consider is the search timestamp pattern. 
By assuming a dedicated fixed-bandwidth communication channel between \client and \server (explained in~\cref{subsec:contribution}), we are able to eliminate the above leakage patterns simultaneously by augmenting \system with fixed-rate bucket retrievals. 
We delay the formal security guarantees and first introduce the complete protocol. 

To improve efficiency, we preprocess the (sorted) data store in a way that $\db$ will be divided into buckets of the same size $Z$, where each bucket will be associated with a tag denoting the range of data it contains. 
\client will maintain tags (and other states for frequency smoothing) for buckets rather than individual search keys, resulting in significant storage space savings of $Z\times$.
The only marginal modification to the query procedure, as depicted in~\cref{alg:range2kv}, is that \client has to partition a range query into a sequence of bucket requests according to the tags. 
The proxy will collect necessary buckets, filter false positives, and return the correct results.  

\begin{algorithm}[t]
    \caption{Bucketization}\label{alg:range2kv}
    \small
    \begin{algorithmic}[1] 
            \Statex\underline{$\mathsf{Bucketize}(\mathsf{arr}, Z)$:} \Comment{For \server:} 
            \State $n\gets |\mathsf{arr}|, B\gets \lceil n/Z\rceil$ 
            \State $\mathsf{arr}.\mathsf{add}(\infty, \dots, \infty)$ 
            \Comment{\footnotesize append $B\cdot Z-n$ dummies}
            \State Split $\mathsf{arr}$ into $B$ buckets of size $Z$ as $\mathsf{bkt}$
            \State $\tags\gets[], \mathsf{pendingQ}\gets[]$ 
            \For{$i\gets 1,\dots, B$} 
            \State $l_i\gets \mathsf{bkt}[i][1], r_i\gets \mathsf{bkt}[i][Z]$
            \State $\tags.\mathsf{add}([l_i, r_i]), \mathsf{pendingQ}.\mathsf{add}(\emptyset)$ 
            \EndFor 
            \State $\mathsf{w_{shuffle}}\sample \mathbb{Z}^{B},\mathsf{OSort}(\mathsf{bkt}, \mathsf{w_{shuffle}})$
            \State { $\mathsf{bkt}', \pi_f, R\gets\pancake.\mathsf{Init}(\hat{\pi},\allowbreak \mathsf{bkt},\allowbreak \alpha)$}
            \State Insert $\mathsf{bkt}'$ into the backend KV store 
            \State \textbf{send} $\tags, \mathsf{pendingQ},\pi_f, R$ to \client, who invokes $\mathsf{Pool.Setup}$ 
    \end{algorithmic} 
    \dotfill 
    \\
    \begin{minipage}[t]{0.48\linewidth}
       Upon receiving a query:
        \begin{algorithmic}[1] 
            \Statex\underline{$\mathsf{Partition}(q=[l, r])$:}
            \State $b_{l}\gets\max\limits_{\mathsf{t}\in\tags}{\mathsf{t}.r< q.l} + 1$
            \State $b_{r}\gets\min\limits_{\mathsf{t}\in\tags}{\mathsf{t}.l> q.r} - 1$ 
            \If{$b_r<b_l$}
            \State \Return \Comment{\footnotesize no candidate}
            \EndIf
            \State $q.cnt\gets b_{r}- b_{l} +1$
            \State $q.data\gets\emptyset$
            \For{$i\gets b_{l},\dots,  b_{r}$} 
            \State $\mathsf{Pool}.\mathsf{Put}(i)$
            \Comment{\footnotesize \cref{alg:samplingPool}}
            \State $\mathsf{pendingQ}[i].\mathsf{add}(q)$
            \EndFor 
        \end{algorithmic} 
    \end{minipage} 
    \begin{minipage}[t]{0.48\linewidth}
    Upon retrieving $i$-th bucket:
        \begin{algorithmic}[1] 
            \Statex\underline{$\mathsf{Reply}(i, \mathsf{data})$:}
            \For{$q\in\mathsf{pendingQ}[i]$} 
            \State $q.data\xleftarrow{\cup } \mathsf{data}$
            \State $q.cnt\gets q.cnt - 1$
            \If{$q.cnt\gets 0$}
              \State {filter $q.data$ and send it back to the client}
            \EndIf 
            \EndFor
            \State $\mathsf{pendingQ}[i]\gets\emptyset$
        \end{algorithmic} 
    \end{minipage} 
\end{algorithm}

\para{Derive the distribution for each bucket being queried.}After transforming the input dataset into buckets, it is necessary to estimate the bucket access distribution to smooth the frequency of bucket accesses.  
We \revise{adopt an assumption similar to \pancake~\citep{Grubbs2020PancakeFS}} that the client proxy has an estimated distribution $\hat{\pi}$ of the true range query distribution $\pi$\revise{, while \citet{Grubbs2020PancakeFS} also discussed how to obtain or estimate such prior knowledge on the query distribution}.
This means that the client proxy is aware of the probability $p_{[x,y]}$ of querying a range $[x, y]$ where $x$ and $y$ are within the domain $\domain$ and $x \leq y$. 
We also assign $p_{[x, y]}$ with $0$ for $y< x$ for convenience. 
Then the probability of a bucket $b$ that contains data from $l$ to $r$ (inclusive) being contained in a result set $\mathsf{S}$ is trivially given by
\begin{equation}\label{eq:bucketDist}
\Pr[b\in \mathsf{S}] = 1-\sum_{i=1}^{l-1}\sum_{j=1}^{l-1}p_{[i,j]} - \sum_{i=r+1}^{N}\sum_{j=r+1}^{N}p_{[i,j]}.
\end{equation}

To compute the query probability for each bucket efficiently, we derive the cumulative distribution as $p_{_{\leq}[x,y]}=\sum_{i=1}^{x}\sum_{j=1}^{y}p_{[i,j]}$ in $\bigO{N^2}$ time. 
We could hence compute it as 
\begin{align*}
    &\Pr[b\in \mathsf{S}]\\
    =&1-p_{_{\leq}[l-1,l-1]}- (p_{_{\leq}[N,N]} - p_{_{\leq}[r,N]} - p_{_{\leq}[N,r]} + p_{_{\leq}[r,r]}) \\ 
    =&p_{_{\leq}[r,N]} + p_{_{\leq}[N,r]} - p_{_{\leq}[r,r]} - p_{_{\leq}[l-1,l-1]}.
\end{align*}
We also show how to efficiently compute bucket access probabilities in the context of uniform range queries, which is commonly assumed by certain attacks~\cite{Kellaris2016GenericAO, Lacharit2018ImprovedRA, Grubbs2019LearningTR}. 
Specifically, if range queries are sampled uniformly at random, as assumed by certain attacks~\cite{Kellaris2016GenericAO, Lacharit2018ImprovedRA, Grubbs2019LearningTR}, each bucket will be queried with a probability of 
\begin{equation}\label{eq:uniformBucketDist}
    \Pr[b\in\mathsf{S}]=\frac{r(2N-r+1)-l(l-1)}{N(N+1)},
\end{equation}
by substituting $p_{[x,y]} $ with $\frac{2}{N(N+1)}, \forall 1\leq x\leq y \leq N$ since there are $N(N+1)/2$ possible distinct range queries. 

\system ensures ROR-CRDA security, rendering adversaries unable to distinguish between requests to \server as point queries or range queries. 
The formal theorem and its detailed proof are provided in~\myappend.

\para{Efficiency.}$\mathsf{Bucketize}$ runs in $\mathcal{O}({n\allowbreak \cdot \log^2 (n/Z)})$ time due to the oblivious shuffling process. 
Note that there are $\mathcal{O}(n/Z\allowbreak \cdot\allowbreak  \log^2 (n/Z))$ compare-and-swap operations in shuffling, with each swap operation between two buckets taking $\bigO{Z}$ time (\ie, the bucket size). 
The overall running time is obtained by multiplying them together.

$\mathsf{Partition}$ runs in $\bigO{\log(n/Z)}$ as buckets covering the query range could be found by binary searches over $\lceil n/Z\rceil $ buckets. 
Line 5 of the $\mathsf{Reply}$ function dominates its running time due to a linear scan to filter out false positives. 
It runs in $\bigO{\nu\cdot Z}$ where $\nu$ denotes the total number of buckets it queries for. 
Therefore, the total query time complexity is $\bigO{\log(n/Z) +\nu\cdot Z}$. 
We note that a larger bucket size will reduce partition time, but increase filtering time.

\subsection{Differentially Oblivious Dynamization}\label{subsec:dynamization}
There are two general approaches to enable dynamic workloads on existing \system. 
\revise{The first approach is to enable the search index to accommodate newly updated data \emph{in place}.}
However, such solutions suffer from either (1) limited dynamics (\eg, \pancake~\cite{Grubbs2020PancakeFS}, SEAL~\cite{Demertzis2020SEAL}) 
\revise{with an upper bound on the total number of data entries fixed on the setup phase. 
Namely, one can only insert a restricted number of entries or replace outdated entries with new ones without expanding the entire data store;
or (2) expensive, \ie, \emph{super-logarithmic} overhead, such as oblivious search trees~\cite{Roche2016APO, Mishra2018OblixAE} and differentially oblivious variants~\cite{WaghCM16}. 
This is due to the fact that \emph{in-place} updates require both reads to identify the appropriate positions for new values, and writes to record the new values and potentially restructure the index. 
}

We circumvent the above predicament via the other approach, data-structure dynamization~\cite{Bentley79, Bentley1980DecomposableSP, MathieuRYY21}. %
It refers to the process of transforming a \emph{static} data structure into a \emph{dynamic} one that supports arbitrarily intermixed insertions and searches. 
Specifically, insertions are supported by destructing old static components and rebuilding them into a new one. 
\citet{Bentley1980DecomposableSP} proposed a dynamization technique, called $k$\emph{-binomial transform}, to maintain $k$ components at all times of respective sizes $\binom{D_1}{1}, \binom{D_2}{2}, \dots, \binom{D_k}{k}$, 
where $0\leq D_1 < D_2 <\cdots<D_k$. 
The $i$-th component will be empty if $D_i<i$.
Such a unique decomposition is guaranteed to exist. 
The read amplification caused by querying (at most) $k$ components is hence independent of the data scale. 
An example is given in~\cref{fig:kBinTrans}.

\begin{figure}[b]
    \centering
    \includegraphics[width=\columnwidth]{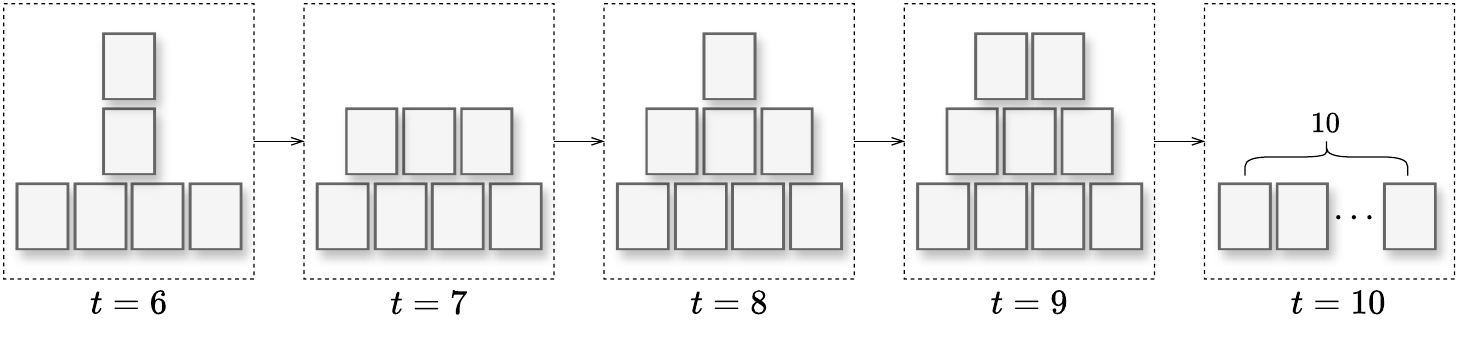}
    \caption{An example of $3$-binomial transform. 
    At each step $t$, layers from top to bottom hold $\binom{D_1}{1}$, $\binom{D_2}{2}$, and $\binom{D_3}{3}$ items, respectively, resulting in a total of $t$ items. 
    We have $D_1=2$, $D_2=3$, and $D_3=4$ for $t=9$. 
    Inserting one more element will trigger the destruction of all three layers and rebuild them into a new one in the third layer, with $D_1=0,\ D_2=1,\ D_3=5$. }
    \label{fig:kBinTrans}
\end{figure}
Since destruction can be easily performed by retrieving bucketized data from the backend storage, the main challenge remains to design a \emph{secure and efficient} rebuild algorithm.
The functionality could be abstracted as \emph{merging} (at most) $k$ sorted arrays into a single sorted one.
The security requires that the rebuilding process should not leak any damaging memory access pattern. 
We hence capture such a demand via differential obliviousness, which enjoys principled privacy-efficiency trade-offs.

\para{Differentially oblivious merge.}\citet{ChanCMS19} proposed an $(\varepsilon, \delta)$-differentially oblivious merge algorithm for two sorted vectors $(\mathsf{a}_0, \mathsf{a}_1)$ in $\mathcal{O}({(|\mathsf{a}_0|+|\mathsf{a}_1|)(\log{\frac{1}{\varepsilon}}+\log\log{\frac{1}{\delta}})})$ time, where neighboring inputs are two pairs of vectors $(\mathsf{a}_0, \mathsf{a}_1)\sim(\mathsf{a}'_0, \mathsf{a}'_1)$ that differ exactly in one element in total. %
Informally, it first allocates two arrays into two lists of bins of the same capacity $\Xi$ in a DO manner, where each bin loads a random number of real elements along with dummies for padding. 
The problem is then transformed into merging two lists of sorted bins and pruning those dummies. 
Intuitively speaking, obliviousness comes from a small \emph{oblivious buffer} when processing two lists, and the noisy number of elements under processing provides differentiality. 
See~\cite{ChanCMS19} for more details.

We note that the size $\Xi$ of the oblivious buffer directly affects performance, as accessing an oblivious buffer incurs a performance overhead of $O(\log \Xi)$.
Meanwhile, $\Xi$ should be large enough to ensure a negligible probability $\delta$ that DO does not hold. 
\citet{ChanCMS19} derived a theoretical lower bound on $\Xi=\Omega(\varepsilon^{-1}\log^5\secpar)$ to guarantee $\delta=\text{exp}({-\Theta{(log^2\secpar})})$, where $\secpar$ is the security parameter. 
In particular, we need to ensure that %
the sum of $B$ i.i.d. truncated Laplace random variables parameterized by $\Xi$ should be greater than $n$ with all but negligible probability. 

In our design, we always allocate at least $n:=Z$ elements into $\hat{B}:=\lceil\frac{2Z}{\Xi(1-\log^{-2}\secpar)}\rceil$ bins. 
Meanwhile, we can compute the distribution of the sum of $\hat{B}$ truncated Laplace random variables via the traditional convolution method.
Subsequently, we can employ a binary search to determine the minimum $\Xi$ that ensures the failure probability is less than $\delta$.
We showcase a notable improvement in reducing $\Xi$ using this numeric method in~\cref{fig:binCapcities}.

\para{$k$-way differentially oblivious merge.}There are two general approaches for merging $k$ sorted arrays in a DO manner: 
\begin{itemize}[leftmargin=*]
    \item 
    \emph{Direct $k$-way merge} imitates the aforementioned DO merge design as follows. 
    We first allocate $k$ arrays into $k$ bin lists; each bin will be tagged with a DP interior point.
    We then leverage a small oblivious buffer to help merge those bins. 
    We could conclude that the buffer should be of size $\bigO{(2k+1)\Xi}$ by following a similar analysis about the number of ``safe'' bins in~\cite{ChanCMS19}. 
    Each bin merging iteration will append elements from a bin to the buffer, obliviously sort the buffer, and remove ``safe'' elements from it. 
    It hence results in $\bigO{n/\Xi\cdot (2k+1)\Xi\cdot \log ((2k+1)\Xi)}$ running time with $n=\sum^{k}_{i=1}|\mathsf{a}_i|$, or $\bigO{kn(\log \frac{1}{\varepsilon} + \log\log \frac{1}{\delta})}$ by replacing $\Xi$ with concrete parameters. One possible optimization is to treat the buffer as an oblivious heap~\cite{Shi2020PathOH} such that inserting $\Xi$ elements in each bin merging iteration will lead to $\bigO{\Xi\cdot \log ((2k+1)\Xi)}$ running time. 
    It hence results in totally $\bigO{n\log{k}(\log \frac{1}{\varepsilon} + \log\log \frac{1}{\delta})}$ time. 
    \item \emph{Iterative merge} applies the DO merge algorithm on pairs of $k$ input arrays. 
    In $i$-th iteration, $(2j-1)$-th array will be merged with the $2j$-th array for $1\leq j\leq k/2$. 
    Each iteration will reduce the number of arrays to half until there is only one array left, which implies total $\lceil\log k\rceil$ iterations. 
    Since each input array will be revolved in one DO merge during each iteration, we obtain the total running time $\bigO{n\log k(\log \frac{1}{\varepsilon} + \log\log \frac{1}{\delta})}$. 
\end{itemize}

We choose the iterative merge as 1) the iterative merge algorithm could be easily parallelized while direct merge cannot; 2) the differential obliviousness for iterative merge could be easily obtained by the composition rule of differential obliviousness~\cite{ChanCMS19, Mingxun2022DOComposition, Zhou2023AdvancedCTDO}.
We notice that every element will be involved in $\log k$ instances of the DO merge algorithm, which hence results in $(\varepsilon\log k, \delta)$-DO. 
However, the DO guarantee of the direct merge algorithm requires in-depth formal analysis, which complicates the situation. 

\begin{algorithm}[t]
\caption{Differentially Oblivious Dynamization in \server}
\label{alg:DODynamizations}
\small

\begin{algorithmic}[1]
\Statex \underline{$\mathsf{Setup}(k)$:} %
\State $D\gets [0, \dots, k-1, \infty]$ \Comment{\footnotesize for $k$-binomial transform}
\State $\mathsf{cmpnt}\gets\emptyset$ 
\end{algorithmic}

\begin{algorithmic}[1]
\Statex \underline{$\mathsf{Update}( \mathsf{records})$:}
\Comment{\footnotesize records in order}
\State $i\gets 0$, $D_i \gets D_i+1$, $\mathsf{arrs}\gets \set{\mathsf{records}}$
\While{$D_i=D_{i+1}$} 
\State $\mathsf{arrs}\gets\mathsf{arrs}\cup \set{\mathsf{cmpnt}_i},  \mathsf{cmpnt}_i\gets \emptyset$
\Statex \Comment{\footnotesize by retrieving buckets based on labels in $\mathsf{cmpnt}_i$}
\State $D_{i+1}\gets D_{i+1}+1, D_i \gets i, i\gets i+1$ 
\EndWhile
\State $\mathsf{arrs}\gets\mathsf{arrs}\cup \set{\mathsf{cmpnt}_i}$
\State $\mathsf{arr}\gets \mathsf{KWayDOMerge}(\mathsf{arrs})$
\State $\mathsf{Bucketize}(\mathsf{arr}, Z)$ \Comment{\footnotesize \cref{alg:range2kv}, obtain $\mathsf{bkts}'$}

\State $\mathsf{cmpnt}_i\gets\mathsf{labels}(\mathsf{bkts}')$\Comment{\footnotesize bucket IDs} 

\end{algorithmic}
\dotfill 
\begin{algorithmic}[1]
    \Statex\underline{$\mathsf{Transform}(\mathsf{idx}, \mathsf{tags}_{new})$:} \label{alg:transform}
    \Comment{\footnotesize Pending query transformation in \client}
    \Statex For components from $1$ to $\mathsf{idx}$
    \For{$t\gets 1, \dots, |\mathsf{tags}_{old}|$} \Comment{\footnotesize $\mathsf{t}$ denotes $\mathsf{tags}_{old}[t]$ }
    \For{$q\in \mathsf{pendingQ}[t]$}
        \State $\mathsf{Partition}\left(\left[\max\left(q.l, \mathsf{t}.l\right), \min\left(q.r, \mathsf{t}.r\right)\right]\right)$ 
    \EndFor 
    \EndFor
\end{algorithmic}
\end{algorithm}

We then present the differentially oblivious dynamization algorithm in~\cref{alg:DODynamizations}.
\system fetches encrypted buckets in the pending components from the KV store and brings them into the enclave for decryption.
It then invokes the iterative $k$-way DO merge algorithm to obtain a sorted array that will be \emph{bucketized} accordingly. 

We notice that \system has to maintain structures introduced in~\cref{subsec:queryDecorrelation} and~\cref{subsec:rangeQuerySupport} on every living component (\ie, $\forall 1\leq i\leq k, i\leq D_i$) to ensure correctness.
It leads to new issues that need to be addressed properly and securely. 
Firstly, we notice that buckets of small components would be accessed more frequently. 
For instance, if two newly inserted records with the minimum and the maximum of the domain $\domain$ compose the only bucket in the first component, then every range query will ask for it.  
To alleviate such performance overhead, the client proxy maintains a local cache of size $Z$ (\ie, the bucket size). 
\client sends a bucket of sorted records together to \server only if the local cache is full. 

Secondly, we need to guarantee the correctness of queries in components to be destroyed that have not yet received responses from all pending buckets. 
Meanwhile, we should minimize the number of targeted buckets in the new component, aiming to introduce minimal performance overhead.
A na\"ive solution is to append pending queries of an old bucket with tag $\mathsf{t}_{old}$ to newly generated buckets with tags $\mathsf{tags}$ by invoking $\mathsf{Partition}(\mathsf{t}_{old})$ in~\cref{alg:range2kv}, \ie, treating the old tag as a query. 
However, such a solution may result in many unnecessary bucket retrievals. 
In the instance above, a bucket with tag $[\min, \max]$ will propagate its pending queries to \emph{all} buckets in the newly generated component, which will download the entire data store if it happens to be the $k$-th component. 
Fortunately, this issue could be easily fixed by invoking $\mathsf{Partition}$ in the interval $([\max(q.l, \mathsf{t}_{old}.l),\allowbreak \min (\allowbreak q.r,\allowbreak  \mathsf{t}_{old}.r ) ] )$, where $q$ is the query pending for future replies. 
We name it pending query transformation and provide a formal description in~\cref{alg:DODynamizations}.

\begin{theorem}[DO$_{\text{update}}$-ODDS]~\label{thm:DOODDS}Let $\varepsilon=\bigO{1}$, $k=\bigO{1}$, $\db$ and $\db'$ be two neighboring data stores, $\ops$ and $\ops'$ be two query-consistent neighboring operational sequences.
Then \system satisfies $(\varepsilon\log k , \delta)$-DO$_{\text{update}}$-ODDS and achieves perfect correctness.
\end{theorem}
We give a formal proof in \myappend\revise{, wherein the argument intuitively follows a $k$-fold sequential composition of DO~\cite{Mingxun2022DOComposition}. 
We remark that \system sustains a consistent level of \emph{per-insertion} privacy loss, irrespective of the cumulative number of insertions into the encrypted data store. 
This contrasts to the log-structured merge tree employed in~\citep{ChanCMS19}, which offers $(\varepsilon\log |\ops|,\delta)$-DO and experiences a \emph{per-insertion} privacy loss that escalates with the cumulative number of inertions performed. 
Furthermore, we can apply a standard sequential composition rule for DO~\cite{Mingxun2022DOComposition} to compute the cumulative privacy loss of all updates to the encrypted data store. 
}\label{revise:R5}

\para{Update and delete supports.}
Updates are performed by inserting new record values with more recent operating timestamps. 
Deletion is indeed a special update with a special record type known as a tombstone. 
Therefore, the client proxy needs to sort the query results according to the search keys and cancel outdated records according to the operating timestamps. 
To further reduce performance overhead, we could mark all outdated records as dummies that will be purged while merging these components.

\revise{
\subsection{Security Assurance and Efficiency Tuning}\label{subsec:discussion}
Following the description of \system design, we proceed to review its security assurance and discuss how to tune efficiency with security. 
As \system evolves to support key-value searches, static range queries, and dynamic range queries, we will also discuss its security considerations in each of these scenarios.  
We emphasize that, throughout our design, we do not consider publishing or sharing private query results, but rather focus on preventing the untrusted server from recovering private client queries via various leakage patterns.

\para{$\theta$-query decorrelation in frequency-smoothed KV stores.}\system aims to mitigate query correlation leakage in frequency-smoothed key-value stores through random sampling. 
It mitigates the \emph{overall} query correlation by eliminating the correlation among \emph{pending} queries. 
This, however, comes at the cost of introducing unstable or increased query latency for users, particularly evident when handling numerous pending queries or enforcing a minimum size on the sampling pool.
The performance overhead is trivial, as the expected times for a query to be sampled uniformly at random from $\theta$ pending queries is $\theta$.
For instance, \system can set $\theta=n$ to align the size of the sampling pool with that of the data store to maximize security. 
Consequently, the trusted client proxy will \emph{randomly} retrieve an object from the untrusted server over the entire data store in each access. 
This ensures that server observations of data accesses remain entirely \emph{unrelated} to the client's real queries.
However, it results in severely restricted utility, since retrieving the target query requires, on average, $n$ random requests to the server.
We also empirically show that $\theta\in[2, 4]$ is adequate to mitigate query correlation leakage in a classic Markov query process. 

\para{Nearly zero-leakage static range queries.}\system then advances to support \emph{static} range queries, offering a strong security guarantee that ensures requests involving range queries to the untrusted server remain indistinguishable from point queries.
Namely, given a sequence of requests, adversaries cannot discern if certain queries constitute a range query.
And it relies on an additional yet reasonable assumption of the presence of a stable communication channel, allowing the trusted proxy to consistently query the untrusted server to conceal sensitive information.
A crucial parameter in this scenario, the bucket size $Z$, affects only the query efficiency.
And its optimization, as discussed in~\cref{subsec:rangeQuerySupport}, involves a delicate balance between partition time and filtering time, depending on the specific characteristics of the workload.

\para{Differentially oblivious updates.}Subsequently, \system expands its functionalities to handle \emph{updates} to the underlying data store.
We remark that to maintain an efficient search index, \system has to arrange the updated data within the index, rather than simply appending it to the end.
However, this "rearrangement" process may inadvertently reveal sensitive access pattern leakage.
Therefore, \system introduces differential obliviousness~\cite{ChanCMS19} (DO) to protect updated values from such leakage. 
Unlike DP, a notion protecting the \emph{values} of individuals, DO employs analogous principles but protects the \emph{access pattern} introduced by individuals, particularly through updates in our scenario. 
Furthermore, we emphasize that the access pattern reveals sensitive but incomplete information about updated data, and adversaries need additional auxiliary information to compromise client privacy~\cite{Kamara2021CryptanalysisOE,Islam2012AccessPD, Oya2021IHOPIS}.
Therefore, one may not need a considerably small $\varepsilon$ to mitigate such leakage. 
In cases where clients do require such a small $\varepsilon$, the efficiency of \system will degrade to an oblivious one, ensuring that no sensitive information about updated data will be leaked. 
Conversely, when aiming for maximal efficiency with a large $\varepsilon$, \system will operate similarly to a nonoblivious one,  disclosing the locations of updated data and potentially revealing their sensitive values.
}

\section{Experimental Evaluation}\label{sec:evaluation}
We implemented \system as a modular client-server application on Intel SGX~\cite{Costan2016IntelSE} in C++.
Our implementation builds upon the Remote Attestation sample code provided with the SGX SDK~\cite{sgxGithub} and utilizes its libraries for encryption, MACs, and hashing.

\subsection{Experimental Setup}\label{subsec:expsetup}
\para{Environment.}We run experiments on a machine with an Intel(R) Xeon(R) Platinum 8369B CPU @ 2.90GHz \revise{of 32 physical cores, with SGXv2 enabled}.
The machine has \revise{128GB RAM, of which about 64GB is enclaves' protected memory.} It operates on Ubuntu 20.04 and uses SGX SDK version 2.19.

\para{Dataset and query.}We use synthetic datasets with varying dataset sizes, record lengths, and domain sizes.
We also generate several query sets with different selectivities (\ie, lengths of the range).
Queries are sampled uniformly from the relevant domain.

\noindent\begin{figure*}
    \centering
    \resizebox{0.8\textwidth}{!}{
    \begin{tikzpicture}

\definecolor{darkgray176}{RGB}{176,176,176}

\begin{axis}[
point meta max=0.0322222222222222,
point meta min=0.0222222222222222,
tick align=outside,
tick pos=left,
title={\pancake with\\correlated queries},
title style={yshift=-6pt, font=\small, align=center}, 
ylabel style={yshift=-2pt},
xlabel style={yshift=6pt},
xticklabel style={rotate=60, anchor=east}, 
at={(0.0\textwidth,0)},
width=0.24\linewidth,
height=0.24\linewidth,
x grid style={darkgray176},
xlabel={To \(\displaystyle q_{i+1}\)},
xmin=-0.5, xmax=5.5,
xtick style={color=black},
xtick={0,1,2,3,4,5},
xticklabels={
  {\(\displaystyle k_{1,0}\)},
  {\(\displaystyle k_{1,1}\)},
  {\(\displaystyle k_{2,0}\)},
  {\(\displaystyle k_{2,1}\)},
  {\(\displaystyle k_{3,0}\)},
  {\(\displaystyle k_{f,0}\)}
},
y dir=reverse,
y grid style={darkgray176},
ylabel={From \(\displaystyle q_i\)},
ymin=-0.5, ymax=5.5,
ytick style={color=black},
ytick={0,1,2,3,4,5},
yticklabels={
  {\(\displaystyle k_{1,0}\)},
  {\(\displaystyle k_{1,1}\)},
  {\(\displaystyle k_{2,0}\)},
  {\(\displaystyle k_{2,1}\)},
  {\(\displaystyle k_{3,0}\)},
  {\(\displaystyle k_{f,0}\)}
}
]
\addplot graphics [includegraphics cmd=\pgfimage,xmin=-0.5, xmax=5.5, ymin=5.5, ymax=-0.5] {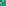};
\end{axis}

\begin{axis}[
point meta max=0.0322222222222222,
point meta min=0.0222222222222222,
tick align=outside,
tick pos=left,
title={\system with\\0-query decorrelation},
title style={yshift=-6pt, font=\small, align=center}, 
ylabel style={yshift=-2pt},
xlabel style={yshift=6pt},
xticklabel style={rotate=60, anchor=east}, 
at={(0.16\linewidth,0)},
width=0.24\linewidth,
height=0.24\linewidth,
x grid style={darkgray176},
xlabel={To \(\displaystyle q_{i+1}\)},
xmin=-0.5, xmax=5.5,
xtick style={color=black},
xtick={0,1,2,3,4,5},
xticklabels={
  {\(\displaystyle k_{1,0}\)},
  {\(\displaystyle k_{1,1}\)},
  {\(\displaystyle k_{2,0}\)},
  {\(\displaystyle k_{2,1}\)},
  {\(\displaystyle k_{3,0}\)},
  {\(\displaystyle k_{f,0}\)}
},
y dir=reverse,
y grid style={darkgray176},
ymin=-0.5, ymax=5.5,
ytick=\empty,
]
\addplot graphics [includegraphics cmd=\pgfimage,xmin=-0.5, xmax=5.5, ymin=5.5, ymax=-0.5] {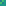};
\end{axis}

\begin{axis}[
point meta max=0.0322222222222222,
point meta min=0.0222222222222222,
tick align=outside,
tick pos=left,
title={\system with\\2-query decorrelation},
title style={yshift=-6pt, font=\small, align=center}, 
ylabel style={yshift=-2pt},
xlabel style={yshift=6pt},
xticklabel style={rotate=60, anchor=east}, 
at={(0.32\textwidth,0)},
width=0.24\linewidth,
height=0.24\linewidth,
x grid style={darkgray176},
xlabel={To \(\displaystyle q_{i+1}\)},
xmin=-0.5, xmax=5.5,
xtick style={color=black},
xtick={0,1,2,3,4,5},
xticklabels={
  {\(\displaystyle k_{1,0}\)},
  {\(\displaystyle k_{1,1}\)},
  {\(\displaystyle k_{2,0}\)},
  {\(\displaystyle k_{2,1}\)},
  {\(\displaystyle k_{3,0}\)},
  {\(\displaystyle k_{f,0}\)}
},
y dir=reverse,
y grid style={darkgray176},
ymin=-0.5, ymax=5.5,
ytick=\empty,
]
\addplot graphics [includegraphics cmd=\pgfimage,xmin=-0.5, xmax=5.5, ymin=5.5, ymax=-0.5] {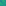};
\end{axis}

\begin{axis}[
point meta max=0.0322222222222222,
point meta min=0.0222222222222222,
tick align=outside,
tick pos=left,
title={\system with\\4-query decorrelation},
title style={yshift=-6pt, font=\small, align=center}, 
ylabel style={yshift=-2pt},
xlabel style={yshift=6pt},
xticklabel style={rotate=60, anchor=east}, 
at={(0.48\textwidth,0)},
width=0.24\linewidth,
height=0.24\linewidth,
x grid style={darkgray176},
xlabel={To \(\displaystyle q_{i+1}\)},
xmin=-0.5, xmax=5.5,
xtick style={color=black},
xtick={0,1,2,3,4,5},
xticklabels={
  {\(\displaystyle k_{1,0}\)},
  {\(\displaystyle k_{1,1}\)},
  {\(\displaystyle k_{2,0}\)},
  {\(\displaystyle k_{2,1}\)},
  {\(\displaystyle k_{3,0}\)},
  {\(\displaystyle k_{f,0}\)}
},
y dir=reverse,
y grid style={darkgray176},
ymin=-0.5, ymax=5.5,
ytick=\empty,
]
\addplot graphics [includegraphics cmd=\pgfimage,xmin=-0.5, xmax=5.5, ymin=5.5, ymax=-0.5] {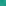};
\end{axis}

\begin{axis}[
colorbar,
colorbar style={ylabel={}},
colormap/viridis,
point meta max=0.0322222222222222,
point meta min=0.0222222222222222,
tick align=outside,
tick pos=left,
title={\pancake with\\independent queries},
title style={yshift=-6pt, font=\small, align=center}, 
ylabel style={yshift=-2pt},
xlabel style={yshift=6pt},
xticklabel style={rotate=60, anchor=east}, 
at={(0.64\textwidth,0)},
width=0.24\linewidth,
height=0.24\linewidth,
x grid style={darkgray176},
xlabel={To \(\displaystyle q_{i+1}\)},
xmin=-0.5, xmax=5.5,
xtick style={color=black},
xtick={0,1,2,3,4,5},
xticklabels={
  {\(\displaystyle k_{1,0}\)},
  {\(\displaystyle k_{1,1}\)},
  {\(\displaystyle k_{2,0}\)},
  {\(\displaystyle k_{2,1}\)},
  {\(\displaystyle k_{3,0}\)},
  {\(\displaystyle k_{f,0}\)}
},
y dir=reverse,
y grid style={darkgray176},
ymin=-0.5, ymax=5.5,
ytick=\empty, 
]
\addplot graphics [includegraphics cmd=\pgfimage,xmin=-0.5, xmax=5.5, ymin=5.5, ymax=-0.5] {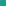};
\end{axis}
\end{tikzpicture}
    }
    \caption{
    Observed query transition frequencies following the Markov process in~\cref{fig:markovModel}.
    \revise{%
    Reduced color variation signifies a decrease in observed query correlation.
    The rightmost figure represents a scenario in which queries are generated independently.%
    }
    }
    \label{fig:decorrelation}
\end{figure*}

\para{Default parameters.}Unless otherwise specified, a number of parameters are set to default values.
The default privacy parameters are $\varepsilon=1.0$, $\lambda=512$ (\ie, $\delta=e^{-\log^2{\lambda}}\approx 10^{-35}$), $\theta=5$ and storage overhead $\alpha=2$ for frequency smoothing.
The default sampling strategy is to sample with equal and constant weights.
We choose a default bucket capacity $Z=512$.
The default data store consists of $n=10^6$ uniformly sampled records of size $4$ KB from the domain $[1,10^6]$.
Queries are uniformly sampled with selectivity $\sigma=0.5\%$, and the optimal batch size for frequency smoothing is estimated as $\lceil 3\cdot n\sigma/Z\rceil$ according to \citet{Grubbs2020PancakeFS}.
We use Redis~\citep{redis} as the default backend key-value store, and \system is instantiated with $k=8$ and parallelized with 32 threads (for $k$-way DO merge and bitonic sorting). And the effective bandwidth is $100$ Mbps.

\begin{figure}[hbt]
    \centering
    \begin{tikzpicture}[
            > = stealth',
            auto,
            prob/.style = {inner sep=1pt,font=\footnotesize},
            transform shape,
            scale=0.8
        ]
        \node[state] (a) {$k_2$};
        \node[state] (b) [above right=of a] {$k_1$};
        \node[state] (c) [below right=of b] {$k_3$};
        \path[->] (a) edge [bend left=10] node[prob]{$0.90$} (b)
        edge [bend left=10] node[prob]{$0.10$} (c)
        (b) edge[loop above] node[prob]{$0.3$} (b)
        edge [bend left=10] node[prob]{$0.65$} (a)
        edge [bend left=10] node[prob]{$0.05$} (c)
        (c) edge [bend left=10] node[prob]{$0.7$} (b)
        edge [bend left=10] node[prob]{$0.3$} (a);
    \end{tikzpicture}
    \begin{tikzpicture}
\definecolor{cornflowerblue107174214}{RGB}{107,174,214}
\definecolor{darkgray176}{RGB}{176,176,176}
\definecolor{lightsteelblue158202225}{RGB}{158,202,225}
\definecolor{steelblue49130189}{RGB}{49,130,189}

\begin{axis}[
tick align=inside,
tick pos=left,
xmin=-0.54, xmax=2.54,
xtick style={color=black, opacity=0},
xtick={0,1,2},
xticklabels={
  \(\displaystyle \mathit{k}_1\),
  \(\displaystyle \mathit{k}_2\),
  \(\displaystyle \mathit{k}_3\)
},
ymin=0, ymax=0.581999607967555,
ytick style={color=black},
xticklabel style = {font=\normalsize},
yticklabel style = {font=\footnotesize},
axis y line=left,
axis x line=bottom,
axis line style={-},
width=0.45\linewidth
]
\draw[draw=none,fill=steelblue49130189] (axis cs:-0.4,0) rectangle (axis cs:0.4,0.554285340921481);
\draw[draw=none,fill=cornflowerblue107174214] (axis cs:0.6,0) rectangle (axis cs:1.4,0.380000409315924);
\draw[draw=none,fill=lightsteelblue158202225] (axis cs:1.6,0) rectangle (axis cs:2.4,0.065714249762595);
\end{axis}

\end{tikzpicture}
    \caption{Markov model (left) and its stationary distribution of queried keywords (right).}
    \label{fig:markovModel}
\end{figure}
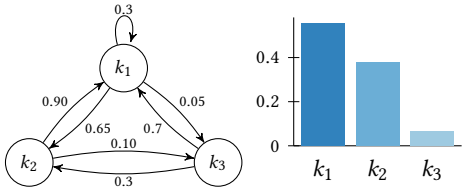

\para{Compared approaches.}\revise{We set up two sets of experiments for comparison with the baseline approaches: one in the context of key-value stores, aligning with \pancake's specialization (corresponding to \system in~\cref{subsec:queryDecorrelation}), and the other in a dynamic range query scenario.

In the first setting, we focus on the efficacy of query decorrelation and the performance overhead introduced by \system relative to \pancake. 
For a clearer demonstration of query decorrelation efficacy, we adopt a simple and classic assumption that client queries follow a Markov process, \ie, the current query depends only on the previous one. 
An adversary can then recover encrypted queries by analyzing the transition frequencies among those queries and the (esitmated) Markov transition probabilities~\cite{Oya2021IHOPIS}.

In light of the above, achieving uniform or smoothed transition frequencies helps in thwarting such attacks, as \emph{uniform} transition frequencies indicate no query correlation.
We will hence measure the efficacy of query decorrelation via two approaches: 1) heat maps, where different colors denote different transition frequencies, and reduced color variation hence signifies a decrease in observed query correlation; and 2) relative standard deviation (RSD), quantifying the degree of dispersion in observed frequencies relative to the uniform, where a smaller RSD implies less observed query correlation.

As outlined in~\cref{subsec:discussion}, the primary efficiency overhead introduced by \system over \pancake is its variable or increased latency, measured in batches required to retrieve the target object.
Additionally, as \pancake duplicates $n$ key-value pairs into $2n$ replicas, we limit $n$ to $3$ for a clearer presentation in the heat maps. 
The chosen Markov model and the stationary distribution of keys are shown in~\cref{fig:markovModel}.
}

We also benchmark the full version of \system against the following approaches with security ranked from weak to strong.

    \noindent\underline{Encryption-only encrypted databases.} Current commercial products (\eg, Azure Always-encrypted~\cite{Antonopoulos2020AzureSD}, Operon~\cite{Wang2022OperonAE}) ignoring leakage mitigation represent the most efficient but least private and secure solution. We take StealthDB~\cite{Gribov2019StealthDBAS} as an example.
    
    \noindent\underline{ObliDB.} It represents oblivious solutions that fully eliminate access pattern but ignore the volume pattern~\cite{Eskandarian2019ObliDBOQ}.
          We estimate the query latency via the StealthDB without index as ObliDB linearly scans the encrypted data store at least twice to answer range queries.
          
    \noindent\underline{\epsolute.} \epsolute combines differential privacy and ORAM to mitigate the volume pattern and fully conceal the memory access pattern, respectively.
          Upon receiving a range query, the user (or a trusted client proxy \client similar to our setting) identifies the true IDs of records targeted by the query via the local indices, and the server \server uses a DP sanitizer to compute the noisy number of records $c$ for \client.
          \client then prepares fake IDs based on $c$ and sends both true and fake IDs to \server to retrieve corresponding records through an ORAM protocol.
          The above solution could still be prohibitively slow in practice.
          Fortunately, the authors bring it to real-world requirements by initiating parallel local indices and ORAMs.
          Similar design ideas are also present in \textsc{ShortStack}~\citep{Vuppalapati2022SHORTSTACKDF}, a distributed and optimized \pancake design.
          Note that we can also draw upon such concepts to scale \system horizontally.
          Therefore, we limit \epsolute with a single ORAM for fairness.

    \noindent\revise{\underline{\system.} We place \system here to mark its more comprehensive security guarantees in comparison to the above three baselines.
    }
    
    \noindent\underline{Linear scan.} It serves as the most private and secure albeit inefficient baseline that \client will download the entire encrypted data store from \server, decrypt, and linearly scan to answer every query.
      It eliminates all five detriment leakage patterns on which we focus.
      We further enable parallel download for a fair comparison with \system.
\subsection{Experiment Stages}
\label{subsec:exps}
\revise{\para{Query decorrelation and performance compared to \pancake.}
}%
As shown in~\cref{fig:decorrelation}, the direct application of our method, even without the minimum requirement on the size of the sampling pool (\ie, $\theta=0$), will significantly smooth the transition frequencies \revise{compared to those of the \pancake design with correlated queries.
When $\theta=4$, the transition frequencies of our proposed approach are close to the uniform ones, \ie, independent queries with no correlation.}
The results imply that \system does effectively decorrelate client queries, thereby thwarting frequency analysis attacks~\cite{Oya2021IHOPIS,Kamara2023Maple}.

To show the tunable trade-offs between privacy and efficiency of \system, we evaluate \system across different sampling pool sizes $\theta$ \revise{and weight update policies.
Specifically, we employ the relative standard deviation (RSD) metric to quantify privacy benefits and measure the efficiency of \system in terms of batches required to retrieve the target data as discussed in~\cref{subsec:expsetup}.}
\Cref{fig:RSDvsLatency} shows the RSD and latency of \system and \pancake with different settings.

In terms of privacy, we find that when \system updates the sampling weight linearly or constantly, the relative standard deviation decreases rapidly and becomes comparable to that of the independent case at $\theta=4$, showing better privacy benefits.
On the other hand, when the sampling weights are updated exponentially, the RSD stabilizes after dropping to 4\%.
This is because, in this case, the sampling pool behaves akin to a FIFO queue, where correlation among pending queries is more likely to be preserved.
Nevertheless, \system still outperforms \pancake, whose RSD is greater than $6\%$.

It is also evident that the aforementioned privacy benefits come at the expense of unstable or increased latency.
The average latency increases from $1.7\times$ to $4\times$ as $\theta$ increases from $1$ to $4$\revise{, consistent with our analysis in~\cref{subsec:discussion}}.
We also notice that the faster the sampling weights grow over time, the more closely the distribution of latency behaves to a FIFO queue, which also confirms our viewpoint.

\begin{figure}[b]
    \resizebox{\linewidth}{!}{%
        \input{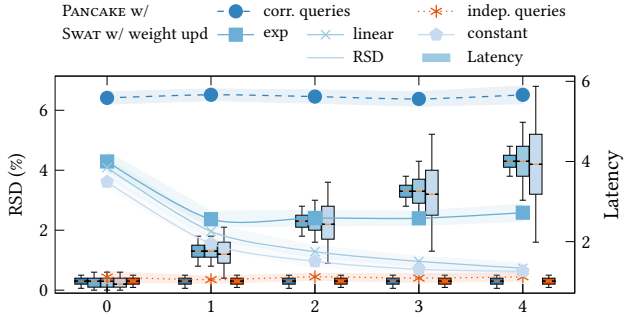}
    }
    \caption{
        Privacy benefits and efficiency penalties over varying sampling pool sizes $\theta$ \revise{and weight update policies against \pancake}.
        RSD \revise{lines} denotes relative standard deviation that measures the degree of dispersion of frequencies relative to the uniform.
        Latency \revise{bars}, measured by the number of batches to retrieve the target data, serves as a key performance indicator to reflect system responsiveness.
        \revise{
        The three middle bars from left to right represent exponential, linear, and constant updates to sampling weights in each batch.
        }
        }\label{fig:RSDvsLatency}
\end{figure}

\noindent\revise{
We will then move on to the evaluation of the full version of \system. 
}

\para{Storage and bandwidth costs.}
The storage usage of two baseline approaches and \system is presented in~\cref{tab:storage}.
We can observe that \client only needs to maintain a \emph{tiny} local state that is $161\times$ smaller than \epsolute.
Meanwhile, the server storage is approximately twice that of the encryption-only one, and is more competitive than \epsolute.
\begin{table}[!t]
    \centering
\resizebox{0.85\linewidth}{!}{
    \begin{minipage}[l]{0.45\linewidth}
        \centering
        \begin{tabular}{l *{2}{c}}
            \toprule
                      & \client & \server \\
                      & MB      & GB      \\
            \midrule
            StealthDB & 0       & $3.81$  \\
            \epsolute & $28.9$  & $12.0$  \\
            \system   & $0.18$  & $7.65$  \\
            \bottomrule
        \end{tabular}
        \caption{Storage usage in default setting.}
        \label{tab:storage}
    \end{minipage}
    \hspace{5pt}
    \begin{minipage}[r]{0.45\linewidth}
        \centering
        \begin{tabular}{l *{3}{c}}
            \toprule
            \diagbox{$\sigma$}{$n$} & $10^5$ & $10^6$ & $10^7$ \\
            \midrule
            0.1\%                   & 6      & 12     & 118    \\
            0.2\%                   & 6      & 24     & 236    \\
            0.5\%                   & 8      & 58     & 586    \\
            \bottomrule
        \end{tabular}
        \caption{Bandwidth costs (MB) with various $\sigma$ and $n$.}
        \label{tab:bandwidth}
    \end{minipage}
    }
\end{table}

Query bandwidth costs in \system are also reasonable as shown in~\cref{tab:bandwidth}.
In specific, a dedicated 1 Gbps network channel can support querying less than 0.1\% of the data from a database with 10 million 4 KB records (equivalent to 40 GB) per second.

\para{Performance against baseline approaches.}
We run experiments with default parameters on \system and the aforementioned baselines (\cref{fig:performanceComparsion}).
StealthDB, configured for (approximately) maximum performance and no leakage mitigation, completes queries within one second, which is just $10.6\times$ faster than \system.
It also implies the necessary time costs associated with encryption/decryption, enclave operations, and network transmission.
Linear scan illustrates the efficiency and practicality of \system compared to the trivially zero-leakage solution.
The difference is $31.6\times$ for the default setting.
In addition, the query latency is competitive against the other two baselines that hide only \emph{subsets of} the leakage patterns we identified.
Also, \epsolute lacks support for dynamic workloads.

\begin{figure}[b]
    \resizebox{0.9\linewidth}{!}{
    \begin{tikzpicture}

\definecolor{cornflowerblue107174214}{RGB}{107,174,214}
\definecolor{darkgray176}{RGB}{176,176,176}
\definecolor{lightsteelblue158202225}{RGB}{158,202,225}
\definecolor{orangered2308513}{RGB}{230,85,13}
\definecolor{powderblue198219239}{RGB}{198,219,239}
\definecolor{steelblue49130189}{RGB}{49,130,189}

\begin{axis}[
log basis y={10},
tick pos=left,
xmin=-0.64, xmax=4.64,
xtick style={draw=none},
xtick={0,1,2,3,4},
xticklabels={StealthDB,ObliDB,\(\displaystyle \mathcal{E}\)psolute,\system, \ \ Linear scan},
yticklabel style={font=\footnotesize},
xticklabel style={font=\small},
y grid style={darkgray176},
ylabel={Query latency},
ymin=0.695477437801663, ymax=415.60514301318,
ymode=log,
axis y line=left,
axis x line=bottom,
ytick style={color=black},
ytick={0.01,0.1,1,10,100,1000,10000},
axis line style={-},
yticklabels={
  \(\displaystyle {10^{-2}}\),
  \(\displaystyle {10^{-1}}\),
  \(\displaystyle {10^{0}}\),
  \(\displaystyle {10^{1}}\),
  \(\displaystyle {10^{2}}\),
  \(\displaystyle {10^{3}}\),
  \(\displaystyle {10^{4}}\)
},
tick align=inside,
tick style={
    major tick length=3pt,
},
width=\linewidth,
height=0.4\linewidth,
]
\draw[draw=none,fill=steelblue49130189,fill opacity=0.5] (axis cs:-0.4,0.7) rectangle (axis cs:0.4,0.93);
\draw[draw=none,fill=cornflowerblue107174214,fill opacity=0.5] (axis cs:0.6,0.6954774) rectangle (axis cs:1.4,10.7);
\draw[draw=none,fill=lightsteelblue158202225,fill opacity=0.5] (axis cs:1.6,0.6954774) rectangle (axis cs:2.4,17.9);
\draw[draw=none,fill=powderblue198219239,fill opacity=0.5] (axis cs:3.6,0.6954774) rectangle (axis cs:4.4,310.8);
\draw[draw=none,fill=orangered2308513,fill opacity=0.5] (axis cs:2.6,0.6954774) rectangle (axis cs:3.4,9.87);
\draw (axis cs:0,0.93) node[
  scale=0.6,
  anchor=south,
  text=black,
  rotate=0.0,
  font=\Large
]{0.93 s};
\draw (axis cs:1,10.7) node[
  scale=0.6,
  anchor=south,
  text=black,
  rotate=0.0,
  font=\Large
]{10.7 s};
\draw (axis cs:2,17.9) node[
  scale=0.6,
  anchor=south,
  text=black,
  rotate=0.0,
  font=\Large
]{17.9 s};
\draw (axis cs:3,9.87) node[
  scale=0.6,
  anchor=south,
  text=black,
  rotate=0.0,
  font=\Large
]{9.87 s};
\draw (axis cs:4,100.8) node[
  scale=0.6,
  anchor=south,
  text=black,
  rotate=0.0,
  font=\Large
]{5.2 min};
\end{axis}

\end{tikzpicture}
    }
    \caption{\emph{Range-query} systems under the default setting.
    The security strength increases from left to right.}\label{fig:performanceComparsion}
\end{figure}
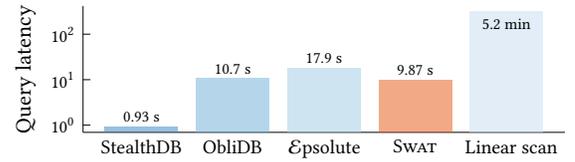

\para{Performance with varying parameters.}We have varied many configuration parameters to measure and understand the impact of them on the performance of \system, including the search time efficiency (in~\cref{fig:queryLatency}) and the update time efficiency (in~\cref{fig:updateTime}).

\noindent\emph{Bucket capacity (\cref{subfig:queryLatencyVSBucketCapacity}).} We observe a notable decrease in query latency from a small $Z$ to a moderate \revise{one}, while the latency gradually increases for larger $Z$.
\revise{A small bucket size performs worse than moderate sizes due to 1) more buckets the dataset being partitioned into and hence more time to locate target buckets, and 2) more batches to fetch the same amount of data.
Conversely, a large bucket size leads to more false positives being transmitted and filtered by the client proxy, thereby resulting in higher latency.
For instance, increasing the bucket size from $2^8$ to $2^9$ reduces the (expected) number of batches (\ie, roundtrips) required to fulfill a range query from $4$ to $2$, thus effectively reducing query latency. 
While further increasing the bucket size to $2^{11}$ increases false positive rates from $24.8\%$ to $39.0\%$, hence increasing filtering time.
}

\noindent\emph{Data, record and result size (\cref{subfig:searchLatencyVSNum,subfig:searchLatencyVSRecordLen,subfig:searchLatencyVSSel}).} We observe query overheads against all three parameters are positive but non-linear.
This is explained by the discrete batch sizes and bucket sampling for query decorrelation.

\begin{figure}[tb]
    \captionsetup[subfigure]{justification=centering}
    \centering
    
    \resizebox{1.0\linewidth}{!}{
    \begin{subfigure}[b]{\linewidth}
        \centering
        \begin{tikzpicture}

\definecolor{darkgray176}{RGB}{176,176,176}
\definecolor{steelblue49130189}{RGB}{49,130,189}

\begin{axis}[
log basis x={2},
tick pos=left,
axis y line=left,
axis x line=bottom,
axis line style={-},
x grid style={darkgray176},
xmin=228.160240162239, xmax=2872.36724301302,
xmode=log,
xtick style={color=black, opacity=0},
ticklabel style={font=\footnotesize},
xtick={64,128,256,512,1024,2048,4096,8192},
xticklabels={
  \(\displaystyle {2^{6}}\),
  \(\displaystyle {2^{7}}\),
  \(\displaystyle {2^{8}}\),
  \(\displaystyle {2^{9}}\),
  \(\displaystyle {2^{10}}\),
  \(\displaystyle {2^{11}}\),
  \(\displaystyle {2^{12}}\),
  \(\displaystyle {2^{13}}\)
},
tick style={
    major tick length=2pt, %
    minor tick length=1pt, %
},
y grid style={darkgray176},
ylabel style={align=center}, ylabel=Latency (s), 
tick align=inside,
ymin=8, ymax=24,
ytick style={color=black},
width=0.95\linewidth,
height=0.3\linewidth
]
\addplot+ [smooth, thick, steelblue49130189, opacity=0.5, mark=*, mark size=2, mark options={solid}]
table {%
256 22.9541048888889
512 10.4753808888889
768 10.15217
1024 10.5864732222222
1280 11.7563995555556
1536 11.9839282222222
1792 12.8617732
2048 13.6361406666667
2304 15.8761813333333
2560 19.34211222
};
\end{axis}

\end{tikzpicture}
        \caption{Bucket capacity}\label{subfig:queryLatencyVSBucketCapacity}
    \end{subfigure}
    }
    \resizebox{0.9\linewidth}{!}{
    \begin{subfigure}[b]{0.32\linewidth}
        \begin{tikzpicture}

\definecolor{cornflowerblue107174214}{RGB}{107,174,214}
\definecolor{darkgray176}{RGB}{176,176,176}
\definecolor{lightsteelblue158202225}{RGB}{158,202,225}
\definecolor{steelblue49130189}{RGB}{49,130,189}

\begin{axis}[
log basis y={10},
tick align=inside,
tick pos=left,
x grid style={darkgray176},
xmin=0.4, xmax=4.7,
xtick style={color=black, opacity=0},
xtick={1.2,2.6,4},
xticklabels={\(\displaystyle 10^4\),\(\displaystyle 10^5\),\(\displaystyle 10^6\)},
y grid style={darkgray176},
ylabel style={align=center},
ylabel=Latency (s),
ymin=0.884529730355192, ymax=11.7838401212584,
ymode=log,
ytick style={color=black, opacity=0},
ytick={0.01,0.1,1,10,100,1000},
yticklabels={
  \(\displaystyle {10^{-2}}\),
  \(\displaystyle {10^{-1}}\),
  \(\displaystyle {10^{0}}\),
  \(\displaystyle {10^{1}}\),
  \(\displaystyle {10^{2}}\),
  \(\displaystyle {10^{3}}\)
},
ticklabel style={font=\footnotesize},
yticklabel style={xshift=4pt}, 
xticklabel style={yshift=2pt},
axis y line=left,
axis x line=bottom,
axis line style={-},
width=1.15\linewidth,
height=0.95\linewidth
]
\draw[draw=steelblue49130189,fill=steelblue49130189,opacity=0.5] (axis cs:0.6,0.001) rectangle (axis cs:1.8,0.9950146);
\draw[draw=cornflowerblue107174214,fill=cornflowerblue107174214,opacity=0.5] (axis cs:2,0.001) rectangle (axis cs:3.2,2.60141577777778);
\draw[draw=lightsteelblue158202225,fill=lightsteelblue158202225,opacity=0.5] (axis cs:3.4,0.001) rectangle (axis cs:4.6,10.4753808888889);
\end{axis}

\end{tikzpicture}
        \caption{Dataset size}\label{subfig:searchLatencyVSNum}
    \end{subfigure}
    \begin{subfigure}[b]{0.32\linewidth}
        \begin{tikzpicture}

\definecolor{cornflowerblue107174214}{RGB}{107,174,214}
\definecolor{darkgray176}{RGB}{176,176,176}
\definecolor{lightsteelblue158202225}{RGB}{158,202,225}
\definecolor{steelblue49130189}{RGB}{49,130,189}

\begin{axis}[
log basis y={10},
tick align=inside,
tick pos=left,
x grid style={darkgray176},
xmin=0.4, xmax=10.5,
xtick style={color=black},
xtick={2.1,5.5,8.9},
xticklabels={8B,128B,4KB},
y grid style={darkgray176},
ymin=0.274190116970388, ymax=12.4597341692746,
ymode=log,
ytick style={color=black},
ytick={0.01,0.1,1,10,100,1000},
yticklabels={
  \(\displaystyle {10^{-2}}\),
  \(\displaystyle {10^{-1}}\),
  \(\displaystyle {10^{0}}\),
  \(\displaystyle {10^{1}}\),
  \(\displaystyle {10^{2}}\),
  \(\displaystyle {10^{3}}\)
},
tick style={opacity=0},
ticklabel style={font=\footnotesize},
xticklabel style={yshift=2pt},
axis y line=left,
axis x line=bottom,
axis line style={-},
yticklabel style={xshift=4pt}, 
width=1.15\linewidth,
height=0.95\linewidth
]
\draw[draw=steelblue49130189,fill=steelblue49130189,opacity=0.5] (axis cs:0.6,0.001) rectangle (axis cs:3.6,0.32613);
\draw[draw=cornflowerblue107174214,fill=cornflowerblue107174214,opacity=0.5] (axis cs:4,0.001) rectangle (axis cs:7,0.66379);
\draw[draw=lightsteelblue158202225,fill=lightsteelblue158202225,opacity=0.5] (axis cs:7.4,0.001) rectangle (axis cs:10.4,10.4753808888889);
\end{axis}

\end{tikzpicture}
        \caption{Record size}\label{subfig:searchLatencyVSRecordLen}
    \end{subfigure}
    \begin{subfigure}[b]{0.32\linewidth}
        \begin{tikzpicture}

\definecolor{cornflowerblue107174214}{RGB}{107,174,214}
\definecolor{darkgray176}{RGB}{176,176,176}
\definecolor{lightsteelblue158202225}{RGB}{158,202,225}
\definecolor{steelblue49130189}{RGB}{49,130,189}

\begin{axis}[
tick align=inside,
tick pos=left,
x grid style={darkgray176},
xmin=0.4, xmax=4.7,
xtick style={color=black},
xtick={1.2,2.6,4},
xticklabels={0.25,0.5,1},
y grid style={darkgray176},
ymin=8, ymax=11.1,
ytick style={color=black},
tick style={opacity=0},
ticklabel style={font=\footnotesize},
xticklabel style={yshift=2pt},
axis y line=left,
axis x line=bottom,
axis line style={-},
yticklabel style={xshift=4pt}, 
width=1.15\linewidth,
height=0.95\linewidth
]
\draw[draw=steelblue49130189,fill=steelblue49130189,opacity=0.5] (axis cs:0.6,0) rectangle (axis cs:1.8,8.29680392592593);
\draw[draw=cornflowerblue107174214,fill=cornflowerblue107174214,opacity=0.5] (axis cs:2,0) rectangle (axis cs:3.2,8.84144816666667);
\draw[draw=lightsteelblue158202225,fill=lightsteelblue158202225,opacity=0.5] (axis cs:3.4,0) rectangle (axis cs:4.6,11.0070785333333);
\end{axis}

\end{tikzpicture}
        \caption{Selectivity (\%)}\label{subfig:searchLatencyVSSel}
    \end{subfigure}
    }
    \caption{Query latency under different settings.}\label{fig:queryLatency}
\end{figure}
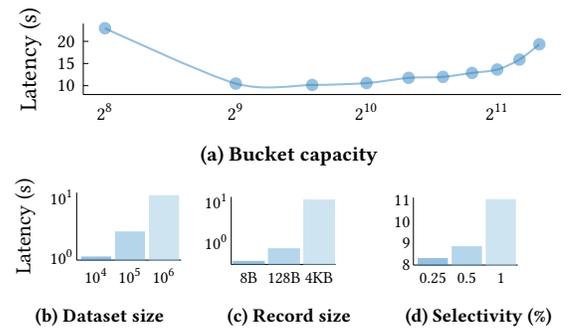

\begin{figure}[b]
    \centering
    
    \resizebox{0.9\linewidth}{!}{
    \begin{subfigure}[b]{0.58\linewidth}
        \centering
        \begin{tikzpicture}

\definecolor{darkgray176}{RGB}{176,176,176}
\definecolor{lightgray204}{RGB}{204,204,204}
\definecolor{lightsteelblue158202225}{RGB}{158,202,225}
\definecolor{steelblue49130189}{RGB}{49,130,189}

\begin{axis}[
legend cell align={left},
legend style={
  fill opacity=0,
  draw opacity=0,
  text opacity=1,
  at={(-0.01,1.12)},
  anchor=north west,
  draw=lightgray204,
  font=\footnotesize,
  legend image post style={scale=0.5}
},
log basis x={10},
log basis y={10},
tick pos=left,
x grid style={darkgray176},
xmin=630.957344480193, xmax=15848931.9246111,
xmode=log,
xtick style={color=black, opacity=0},
xtick={10,100,1000,10000,100000,1000000,10000000,100000000,1000000000},
xticklabels={
  \(\displaystyle {10^{1}}\),
  \(\displaystyle {10^{2}}\),
  \(\displaystyle {10^{3}}\),
  \(\displaystyle {10^{4}}\),
  \(\displaystyle {10^{5}}\),
  \(\displaystyle {10^{6}}\),
  \(\displaystyle {10^{7}}\),
  \(\displaystyle {10^{8}}\),
  \(\displaystyle {10^{9}}\)
},
y grid style={darkgray176},
ylabel={Time (min)},
ylabel style={yshift=-4pt},
ymin=0.001, ymax=1000,
ymode=log,
ytick style={color=black},
ytick={1e-02,1,100,900},
yticklabels={
  \(\displaystyle {10^{-2}}\),
  \(\displaystyle {10^{0}}\),
  \(\displaystyle {10^{2}}\),
},
yticklabel style={xshift=2pt},
tick style={
    major tick length=2pt, %
    minor tick length=1pt, %
},
axis y line=left,
axis x line=bottom,
axis line style={-},
ticklabel style={font=\footnotesize},
tick align=inside,
width=1\linewidth,
height=0.6\linewidth
]
\addplot+ [smooth, thick, steelblue49130189, opacity=0.8, mark=*, mark size=1.5, mark options={solid}]
table {%
1000 0.0029
10000 0.01825
100000 0.1660667
1000000 2.21338
10000000 26.5606
};
\addlegendentry{Bulk insert}
\addplot+ [smooth, thick, lightsteelblue158202225, opacity=0.8, mark=+, mark size=1.5, mark options={solid}]
table {%
1000 0.0027833
10000 0.06813333
100000 1.02701666
1000000 13.1834333
10000000 118.650900
};
\addlegendentry{Sequential insert}
\end{axis}

\end{tikzpicture}
        \caption{Dataset size}\label{subfig:setupTimeVSN}
    \end{subfigure}
    \begin{subfigure}[b]{0.4\linewidth}
        \centering
        \begin{tikzpicture}

\definecolor{darkgray176}{RGB}{176,176,176}
\definecolor{steelblue49130189}{RGB}{49,130,189}

\begin{axis}[
tick align=inside,
tick pos=left,
x grid style={darkgray176},
xmin=-1.1, xmax=67.1,
xtick style={color=black, opacity=0},
xtick={2, 8,16,32,64},
xticklabels={
  \(\displaystyle {2}\),
  \(\displaystyle {8}\),
  \(\displaystyle {16}\),
  \(\displaystyle {32}\),
  \(\displaystyle {64}\),
},
y grid style={darkgray176},
ymin=1.5, ymax=11,
ymode=log,
log basis y={2},
ytick style={color=black},
ytick={1,2,4, 8},
yticklabels={
  \(\displaystyle {1}\),
  \(\displaystyle {2}\),
  \(\displaystyle {4}\),
  \(\displaystyle {8}\),
},
ticklabel style={font=\footnotesize},
tick style={
    major tick length=2pt, %
    minor tick length=1pt, %
},
axis y line=left,
yticklabel style={xshift=2pt},
axis x line=bottom,
axis line style={-},
tick align=inside,
width=1.25\linewidth,
height=0.9\linewidth
]
\addplot+ [smooth, thick, steelblue49130189, opacity=0.8, mark=*, mark size=1.5, mark options={solid}]
table {%
1 10.36535
4 5.80651
8 3.64091666
16 2.573833
32 2.21338333
64 2.0799499
};
\end{axis}

\end{tikzpicture}
        \caption{Thread number}\label{subfig:setupTimeVSThread}
    \end{subfigure}}
    
    \resizebox{0.9\linewidth}{!}{
    \begin{subfigure}[b]{0.49\linewidth}
        \centering
        \begin{tikzpicture}

\definecolor{darkgray176}{RGB}{176,176,176}
\definecolor{steelblue49130189}{RGB}{49,130,189}
\definecolor{lightsteelblue158202225}{RGB}{158,202,225}

\begin{axis}[
tick align=outside,
tick pos=left,
x grid style={darkgray176},
xmin=3.6, xmax=12.4,
xtick style={color=black, opacity=0},
y grid style={darkgray176},
ylabel={Time (min)},
ymin=11.2268916666667, ymax=22.9963416666667,
ytick style={color=black},
ticklabel style={font=\footnotesize},
tick style={
    major tick length=2pt, %
    minor tick length=1pt, %
},
ylabel style={yshift=-4pt},
axis y line=left,
axis x line=bottom,
axis line style={-},
tick align=inside,
yticklabel style={xshift=2pt},
width=1.1\linewidth,
height=0.668\linewidth
]
\addplot+ [smooth, thick, lightsteelblue158202225, opacity=0.8, mark=+, mark size=1.5, mark options={solid}]
table {%
4 22.4613666666667
5 17.15885
6 16.083
7 15.1241166666667
8 13.1834333333333
9 11.9576666666667
10 11.7618666666667
11 12.20355
12 12.7963833333333
};
\end{axis}

\end{tikzpicture}
        \caption{$k$-binomial transformation }\label{subfig:insertionVSK}
    \end{subfigure}
    \begin{subfigure}[b]{0.49\linewidth}
        \centering
        \begin{tikzpicture}

\definecolor{cornflowerblue107174214}{RGB}{107,174,214}
\definecolor{darkgray176}{RGB}{176,176,176}
\definecolor{lightsteelblue158202225}{RGB}{158,202,225}
\definecolor{powderblue198219239}{RGB}{198,219,239}
\definecolor{steelblue49130189}{RGB}{49,130,189}

\begin{axis}[
tick pos=left,
x grid style={darkgray176},
xmin=0.41, xmax=4.59,
xtick style={color=black, opacity=0},
xtick={1,2,3,4},
xticklabels={0.001,0.1,1,10},
y grid style={darkgray176},
ymin=2.15, ymax=2.25,
ytick style={color=black},
ticklabel style={font=\footnotesize},
tick style={
    major tick length=2pt,
    minor tick length=1pt,
},
ytick={2.16,2.2,2.24},
yticklabels={\(\displaystyle 2.16\),
\(\displaystyle 2.20\),
\(\displaystyle 2.24\)},
ylabel style={yshift=-4pt},
axis y line=left,
axis x line=bottom,
axis line style={-},
tick align=inside,
yticklabel style={xshift=2pt},
width=1.2\linewidth,
height=0.668\linewidth
]
\addplot+ [smooth, thick, steelblue49130189, opacity=0.8, mark=*, mark size=1.5, mark options={solid}]
table {%
1 2.23953333333333
2 2.22025
3 2.1951333
4 2.16545
};
\end{axis}

\end{tikzpicture}
        \caption{Privacy budget $\varepsilon$ }\label{subfig:insertionVSEps}
    \end{subfigure}
    }
    \caption{Inserting records under different settings.  Records are inserted all at once in bulk mode and bucket by bucket in sequential mode.}\label{fig:updateTime}
\end{figure}
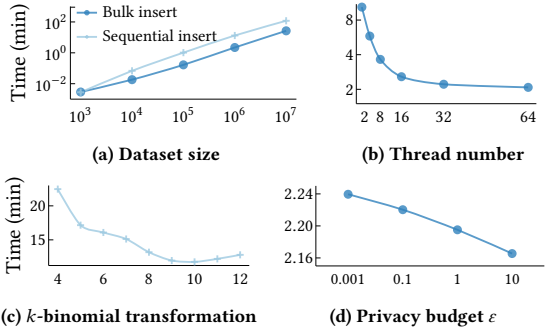

\noindent\emph{Data size and thread num on updates (\cref{subfig:setupTimeVSN,subfig:setupTimeVSThread}).}
We have tried $10^3$ to $10^7$ data insertions in both bulk and sequential modes.
It is evident that the bulk insertion mode outperforms the sequential one by orders due to computation saves on the dynamization process.
Besides, the \emph{logarithmized} running time of both modes increases linearly with the logarithmized data size, which matches the above asymptotic notions.

\noindent\emph{Component number (\cref{subfig:insertionVSK}).}
The total running time for one million sequential insertions drops as the value $k$ increases at the beginning and gradually increases when $k$ exceeds 10.
It conforms to the theoretical write overhead $(k!n)^{1/k}$ that obtains its minimum at $k=10$ when $n=10^6$ (for $k\in\mathbb{N}$).

\noindent\emph{Privacy budget $\varepsilon$ (\cref{subfig:insertionVSEps}).}
$\varepsilon$ determines the oblivious buffer size and \revise{hence} affects the setup time.
A stricter privacy budget demands a larger oblivious buffer, resulting in a longer running time.

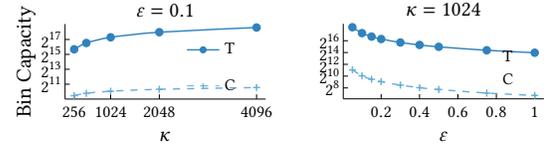
\begin{figure}[!t]
    \centering
    \resizebox{0.9\linewidth}{!}{
    \begin{subfigure}[b]{0.49\linewidth}
    \centering 
    \begin{tikzpicture}

\definecolor{cornflowerblue107174214}{RGB}{107,174,214}
\definecolor{darkgray176}{RGB}{176,176,176}
\definecolor{lightgray204}{RGB}{204,204,204}
\definecolor{steelblue49130189}{RGB}{49,130,189}

\begin{axis}[
legend cell align={left},
legend style={
  fill opacity=0.0,
  draw opacity=0,
  text opacity=1,
  at={(0.91,0.5)},
  anchor=east,
  draw=lightgray204,
  font=\footnotesize,
  legend image post style={scale=0.8}
},
log basis y={2},
tick align=outside,
tick pos=left,
title={\(\displaystyle \varepsilon=0.1\)},
title style={yshift=-6pt},
x grid style={darkgray176},
xlabel={\(\displaystyle \secpar\)},
xmin=64, xmax=4288,
xtick style={color=black},
xtick={256,1024,2048,4096},
xticklabels={
  \(\displaystyle {256}\),
  \(\displaystyle {1024}\),
  \(\displaystyle {2048}\),
  \(\displaystyle {4096}\),
},
y grid style={darkgray176},
ylabel={Bin Capacity},
ymin=512.801882219265, ymax=550000,
ymode=log,
ytick style={color=black},
ytick={2048,8192,32768,131072},
yticklabels={
  \(\displaystyle {2^{11}}\),
  \(\displaystyle {2^{13}}\),
  \(\displaystyle {2^{15}}\),
  \(\displaystyle {2^{17}}\),
},
ticklabel style={font=\footnotesize},
tick style={
    major tick length=2pt, %
    minor tick length=1pt, %
},
ylabel style={yshift=-4pt},
axis y line=left,
axis x line=bottom,
axis line style={-},
tick align=inside,
yticklabel style={xshift=2pt},
width=1.1\linewidth,
height=0.65\linewidth
]
\addplot+ [smooth, semithick, steelblue49130189, mark=*, mark size=1.5, mark options={solid}]
table {%
256 52430
512 94480
1024 160004
2048 257686
4096 398138
};
\addlegendentry{}
\addlegendimage{steelblue49130189, semithick, mark=*, mark size=1.5, opacity=1}
\addlegendentry{T}
\addplot+ [smooth, semithick, cornflowerblue107174214, dashed, mark=+, mark size=1.5, mark options={solid}]
table {%
256 704
512 876
1024 1064
2048 1272
4096 1500
};
\addlegendentry{}
\addlegendimage{cornflowerblue107174214, dashed, mark=+, mark size=1.5, opacity=1}
\addlegendentry{C}
\end{axis}

\end{tikzpicture}
    \end{subfigure}
    \begin{subfigure}[b]{0.49\linewidth}
        \centering 
        \begin{tikzpicture}

\definecolor{cornflowerblue107174214}{RGB}{107,174,214}
\definecolor{darkgray176}{RGB}{176,176,176}
\definecolor{lightgray204}{RGB}{204,204,204}
\definecolor{steelblue49130189}{RGB}{49,130,189}

\begin{axis}[
legend cell align={left},
legend style={
  fill opacity=0.0,
  draw opacity=0,
  text opacity=1,
  at={(0.91,0.4)},
  anchor=east,
  draw=lightgray204,
  font=\footnotesize,
  legend image post style={scale=0.8}
},
log basis y={2},
title style={yshift=-6pt},
tick align=outside,
tick pos=left,
title={\(\displaystyle \secpar = 1024\)},
x grid style={darkgray176},
xlabel={\(\displaystyle \varepsilon\)},
xmin=0.0025, xmax=1.0475,
xtick style={color=black},
y grid style={darkgray176},
ymin=71.0089764261811, ymax=477695.042334019,
ymode=log,
ytick style={color=black},
ytick={16,64,256,1024,4096,16384,65536}, %
yticklabels={
  \(\displaystyle {2^{4}}\),
  \(\displaystyle {2^{6}}\),
  \(\displaystyle {2^{8}}\),
  \(\displaystyle {2^{10}}\),
  \(\displaystyle {2^{12}}\),
  \(\displaystyle {2^{14}}\),
  \(\displaystyle {2^{16}}\),
},
ticklabel style={font=\footnotesize},
tick style={
    major tick length=2pt, %
    minor tick length=1pt, %
},
ylabel style={yshift=-4pt},
axis y line=left,
axis x line=bottom,
axis line style={-},
tick align=inside,
yticklabel style={xshift=2pt},
width=1.1\linewidth,
height=0.65\linewidth
]
\addplot+ [smooth, semithick, steelblue49130189, mark=*, mark size=1.5, mark options={solid}]
table {%
0.05 320006
0.1 160004
0.15 106670
0.2 80002
0.3 53336
0.4 40002
0.5 32002
0.75 21334
1 16002
};
\addlegendentry{T}
\addplot+ [smooth, semithick, cornflowerblue107174214, dashed, mark=+, mark size=1.5, mark options={solid}]
table {%
0.05 2132
0.1 1064
0.15 708
0.2 534
0.3 356
0.4 264
0.5 212
0.75 140
1 106
};
\addlegendentry{C}
\end{axis}

\end{tikzpicture}
    \end{subfigure}
    }
    \caption{Minimal bin capacity obtained through the theoretical analysis (T) ~\cite{ChanCMS19} vs. convolutional computation (C) with different security parameters $\secpar$ and privacy budgets $\varepsilon$.}
    \label{fig:binCapcities}
\end{figure}

\section{Related Work}\label{sec:relatedWork}
CryptDB~\cite{Popa2011CryptDBPC} and the subsequent MONOMI \cite{TuKMZ13} utilize special cryptographic primitives 
to allow different operations over ciphertexts with different security levels.
Arx~\cite{Poddar2016ArxAS} improves the security level by using semantically secure encryption only. 
There are numerous differential privacy data analytics systems, such as Crypt$\epsilon$~\cite{ChowdhuryW0MJ20}, 
Apex~\cite{Ge2019Apex}, PINQ~\cite{McSherry10}, and PrivateSQL~\cite{KotsogiannisTHF19}. 
This line of work assumes a fully trusted data curator (\ie, the cloud service provider) that deviates greatly from our security model. 
\revise{Meanwhile, encrypted databases based on secure multiparty computation and homomorphic encryption are also evolving. 
However, the former designs~\cite{Liagouris23SECRECY,  volgushev2019conclave, ZhangBNM23, BaterHEMR18, wang2022incshrink} require multiple noncolluding servers (which also deviate from our setting), and the latter ones~\cite{Reagen21Cheetah, SamardzicFKDDP021F1, TanLWRA21Efficient, Xuanle2022HEDA, bian2023he3db} are yet to be practically deployed on large-scale datasets.}\label{revise:R6_1}
Another line of work leverages hardware enclaves, such as Cipherbase~\cite{ArasuBEJKKRUV13}, TrustedDB~\cite{Bajaj2014TrustedDBAT}, StealthDB~\cite{Gribov2019StealthDBAS}, EnclaveDB~\cite{Priebe2018EnclaveDBAS}, Enclage~\cite{Sun2021BuildingES}, Operon~\cite{Wang2022OperonAE}.

There is also a line of work focusing on leakage suppression in encrypted databases. 
Opaque~\cite{Zheng2017OpaqueAO}, Oblix~\cite{Mishra2018OblixAE}, Obladi~\cite{Crooks0CHAA18}, POSUP~\cite{HoangOJY19}, HIRB~\cite{Roche2016APO}, and ObliDB \cite{Eskandarian2019ObliDBOQ} focus on fully hiding access pattern via different oblivious designs. 
Given significant overhead, recent work has increasingly concentrated on leakage mitigation that also maintains acceptable performance.
Frequency smoothing~\citep{MavroforakisCOK15, Grubbs2020PancakeFS, Vuppalapati2022SHORTSTACKDF} does not hide which data are accessed but the frequency to be accessed. 
\revise{\citet{maiyya2023waffle} propose a novel security notion to measure the uniformity of accesses to \emph{key-value stores}.
They also developed Waffle that achieve this notion and mitigate both access frequency and query correlation patterns without prior knowledge of data access distribution, offering tunable security-performance trade-offs as well. 
}\label{revise:R6_2}
Differential obliviousness~\citep{Kellaris2016GenericAO, WaghCM16, wagh2018differentially, ChanCMS19, PersianoY19, Persiano2022LowerBF} as another promising direction introduces the differential privacy notion into the protection of memory access patterns. 
Similarly, many works mitigate volume pattern leakage by adding differentially private noise to the result set. 
Coarse-grained obliviousness~\cite{RachidRM20, Ren2020HybrIDXNH, Sun2021BuildingES} relaxes the requirement for indistinguishable distributions of memory blocks to higher granularity (\eg, memory pages). And coarse-grained volume-hiding techniques~\cite{Kamara2018StructuredEA, Mishra2018OblixAE, Demertzis2020SEAL} follow a similar approach by retrieving results in batches (\ie, a fixed number of results), trading correctness for security by returning a fixed number of results per query, or adjustably padding the result set to its nearest power of a user-defined parameter $x$.
In addition, \citet{Wang2021DPSync} considered the update timestamp pattern and mitigated this by updating at a fixed rate or updating only if the number of updates satisfies a noised threshold that satisfies differential privacy.

\section{Conclusion}\label{sec:conclusion}
Systematic leakage mitigation in encrypted data stores has become a critical problem.
In this paper, we present \system, an enclave-assisted encrypted data store that efficiently mitigates various leakages in key-value, range-query, and dynamic workloads. 
We decorrelate queries with an almost-for-free sampling pool to enhance the security of \pancake. 
We implement a system prototype and conduct a comprehensive evaluation on extensive and diverse datasets and workloads, demonstrating excellent performance while meeting the specified security definitions.
\section*{Acknowledgements}
We thank our anonymous reviewers for their insightful feedback. This work is supported in part by the National Key Research and Development Program of China 2023YFB2904000, by the National Natural Science Foundation of China under Grant (NSFC) 62202228, U20A20178, U23A20306, 62072395, 62032021, 62206207, by the Natural Science Foundation of Jiangsu Province under Grant BK20210330, by the Fundamental Research
Funds for the Central Universities 30923011023, and by HK RGC under Grants CityU 11217620, RFS2122-1S04, R6021-20F, R1012-21, C2004-21G, and C1029-22G.  
This work is also partially supported by Alibaba Group through Alibaba Innovative Research Program.

\bibliographystyle{ACM-Reference-Format}
\bibliography{reference}

\section*{Appendix}
\label{sec:security}
We prove the security guarantees that \system can provide in this section.
Let $T$ be query rate, $\pi$ be the true distribution of range queries, $\hat{\pi}$ be the estimated distribution of range queries, 
$m$ be the bit length of bucket key, $\mathcal{C}$ be the ciphertext space of encrypted buckets, $B$ be the batch size. %
We then present the formal definitions used by ROR-CRDA security. 

\para{Pseudorandom function (PRF).} For a keyed function $F: \mathcal{K}\times\bin^*\to \bin^m$ and adversary $\adv$, we define the PRF advantage from the following two games. 
In PRF$^{\adv}_0$, \adv{} has access to an oracle that takes inputs a uniformly random key and a bit string from $\bin^*$ chosen by \adv, and outputs the PRF value evaluated on that point. 
In PRF$^{\adv}_1$, \adv{} has access to an oracle that is a lazy-sampled random function defined over the same space. 
The PRF advantage in distinguishing the above two games is: 
$$
\advantage{prf}{\adv, F}=\left|\Pr[\text{PRF}_0^{\adv}=1]-\Pr[\text{PRF}_1^{\adv}=1]\right|.
$$

\para{Authenticated encryption with associated data (AEAD).}An AEAD scheme $E=(\mathsf{Gen}, \mathsf{Enc}, \mathsf{Dec})$ consists of three algorithms, where $E.\mathsf{Gen}$ takes a security parameter as inputs and outputs elements a key $k\in\mathcal{K}$; $E.\mathsf{Enc}$ takes inputs as a key $k\in\mathcal{K}$, a plaintext $m\in\mathcal{M}$, and optionally additional associated data, outputs ciphertexts $c\in\mathcal{C}$ and an authentication tag; $E.\mathsf{Dec}$ takes inputs as a key $k\in\mathcal{K}$, a ciphertext $c\in\mathcal{C}$, authentication tag, and optionally additional associated data as inputs, and outputs a plaintext $m\in\mathcal{M}$ or $\bot$ (if the ciphertext or associated data does not match the authentication tag). 

For an AEAD scheme $E$ and adversary $\adv$, we define the real-or-random (ROR) advantage from the following two games. 
In ROR$^{\adv}_0$, \adv{} has access to an oracle $E.\mathsf{Enc}$ with a key generated by $E.\mathsf{Gen}$ and plaintexts $m$ chosen by itself. 
In ROR$^{\adv}_1$, \adv{} has access to an oracle that returns uniformly random bit strings of the same length as the ciphertexts in ROR$^{\adv}_0$ . 
The ROR advantage in distinguishing the above two games is: 
$$
\advantage{ror}{\adv, E}=\left|\Pr[\text{ROR}_0^{\adv}=1]-\Pr[\text{ROR}_1^{\adv}=1]\right|.
$$

For a distribution $\pi$ and an adversary $\adv$, who takes inputs elements from the support set of $\pi$ and outputs a bit, we define DIST$_{\pi}^{\adv}$ the game that samples $\pi$ for $T$ times and runs $\adv$ on the resulting sequence. 
For two distributions $\pi,\pi'$ with the same support set, we measure their indistinguishability from an adversary $\adv$, who takes inputs a distribution and outputs a bit, via the following measure: 
$$
\advantage{dist}{\adv, T, \pi, \pi'}=\left|\Pr[\text{DIST}_{T, \pi}^{\adv}=1]-\Pr[\text{DIST}_{T, \hat{\pi}}^{\adv}=1]\right|.
$$

\begin{definition}[Real-or-random indistinguishability under chosen-range-query-distribution attacks]\label{def:ror-crda}
The advantage of an adversary $\adv$ against an encrypted data store $\Pi$ in the real-or-random indistinguishability under chosen-range-distribution-attack (ROR-CRDA) game is defined as follows:

\begin{pchstack} 

\procedure{\textnormal{ROR-CRDA0}$^{\adv}_{T, \pi, \hat{\pi}, \alpha, \bar{t}, Z}$}{
\db\gets\adv_1(\secparam) \\
(\db', \pi_f, \delta)\gets\setup(\db, \alpha, \hat{\pi}, Z) \\ 
\pcfor t\gets 1, 2, \dots, T\pcdo \\ 
\t q\sample \pi \t\Comment{$q:=[l_t, r_t]$}\\ 
\t \Pi.\mathsf{Partition}(q) \\
\t \pcif t \bmod \bar{t} = 0\pcthen \\
\t\t \left(l_1,\dots, l_B\right)\gets \Pi.\mathsf{Batch}() \\
\t\t \pcfor i\gets 1, \dots, B\pcdo \\ 
\t\t\t \tau\left[t\right]\left[i\right]\gets(l_i, \db'\left[l_i\right])\\ 
b\gets\adv_2(\db', \tau)\\
\pcreturn b
}
\procedure{\textnormal{ROR-CRDA1}$^{\adv}_{T, \alpha, \bar{t}, Z}$}{
\db\gets\adv_1(\secparam)\\ 
\db'\gets\emptyset,\mathsf{labels}\gets\emptyset \\ 
n\gets |\db|, \mu\gets\left\lceil\alpha \cdot  \left\lceil n/Z \right\rceil \right\rceil \\ 
\pcfor i\gets 1, \dots, \mu \pcdo \\ 
\t l_i \sample\bin^m, \nu_i \sample \mathcal{C} \\ 
\t \mathsf{labels}\gets\mathsf{labels}\cup \set{l_i} \\
\t \db'\gets\db'\cup\set{(l_i, \nu_i)} \\ 
\pcfor t \gets 1, 2, \dots, T \pcdo \\ 
\t \pcif t \bmod \bar{t} \neq 0\pcthen \\
\t\t \pccontinue \\ 
\t \pcfor i\gets 1,\dots,  B \pcdo \\
\t\t l\sample \mathsf{labels}, \nu\gets\db'[l]\\ 
\t\t \tau[t][i]\gets (l, \nu )\\ 
b\gets\adv_2(\db', \tau)\\
\pcreturn b
}
\end{pchstack}

\begin{align*}
\advantage{\textnormal{ror-crda}}{\adv, \Pi}=&\left|\Pr[\textnormal{ROR-CRDA0}^{\adv}_{T, \pi, \hat{\pi}, \alpha, \bar{t}, Z}=1]\right. \\ 
&\left.-\allowbreak \Pr[\textnormal{ROR-CRDA1}^{\adv}_{T, \alpha, \bar{t}, Z}=1]\right|. 
\end{align*}

\end{definition}
\begin{theorem}
    Let $T\geq 1$, $\hat{\pi}$ be an estimate of the true distribution $\pi$ of range queries. 
    For any probabilistic polynomial-time adversary $\adv$, we have 
    $$
        \advantage{\text{ror-crda}}{\adv, \system}\leq\advantage{prf}{\adv,F}+\advantage{ror}{\adv,E}+\advantage{dist}{\adv,\pi, \hat{\pi}}, 
    $$
    where $F$ and $E$ denote the PRF and AEAD scheme the underlying \pancake uses, and \system adopts the unweighted sampling policy for query decorrelation.
\end{theorem}
\begin{proof} 
We consider a sequence of games:
$$\small
\textsc{G}^{\adv}_{\text{ROR-CRDA0}}\rightarrow \textsc{G}^{\adv}_1\rightarrow \textsc{G}^{\adv}_2\rightarrow \textsc{G}^{\adv}_3\rightarrow \textsc{G}^{\adv}_{\text{ROR-CRDA1}}.
$$
Game $\textsc{G}^{\adv}_1$ is the same as $\textsc{G}^{\adv}_{\text{ROR-CRDA0}}$ except we replace the PRF $F$ for bucket ID generation with a uniformly random function. 
The advantage in distinguishing them is upper bounded by the advantage of a PRF adversary: 
$$
|\Pr[\textsc{G}^{\adv}_{\text{ROR-CRDA0}}=1]-\Pr[\textsc{G}^{\adv}_1=1]|\leq \advantage{prf}{\adv,F}.
$$
Game $\textsc{G}^{\adv}_2$ is the same as $\textsc{G}^{\adv}_1$ except we replace the AEAD scheme $E$ for encrypting buckets with a uniformly random function outputting bit strings in the same ciphertext space. 
The advantage in distinguishing them is upper bounded by the advantage of an AEAD adversary: 
$$
|\Pr[\textsc{G}^{\adv}_{1}=1]-\Pr[\textsc{G}^{\adv}_2=1]|\leq \advantage{ror}{\adv,E}.
$$
Game $\textsc{G}^{\adv}_3$ is the same as $\textsc{G}^{\adv}_2$ except we replace the estimated range distribution $\hat{\pi}$ with the true range distribution $\pi$. 
The advantage could easily be derived as: 
$$
|\Pr[\textsc{G}^{\adv}_{2}=1]-\Pr[\textsc{G}^{\adv}_3=1]|\leq \advantage{dist}{\adv,\pi, \hat{\pi}}.
$$
It is easy to verify that the numbers of random bitstrings for labels and values are the same in $\textsc{G}_3$ and $\textsc{G}_{\text{ROR-CRDA1}}$. 
Therefore, the key argument remaining is that accesses of label-value pairs in $\textsc{G}_3$ are identically distributed to those in $\textsc{G}_{\text{ROR-CRDA1}}$. 
The correctness of this argument relies on the validity of the following three blocks:
\begin{enumerate}
    \item Correct estimation of the bucket access distribution. 
    As shown in~\cref{eq:bucketDist,eq:uniformBucketDist}, we map the known distribution of range queries to the distribution of bucket accesses, whose correctness can be easily verified by considering the complementary event of a range query accessing a bucket (\ie, the query range does not overlap with the bucket data range);
    \item Independence of $\tau[t][i]$. This naturally holds since bucket labels associated with the same range query are independently sampled, and bucket accesses are indistinguishable from truly random ones due to the oblivious shuffling performed during the setup stage.
    \item Uniform access distribution of each bucket replica, as argued in~\cite{Grubbs2020PancakeFS}.
\end{enumerate}
The theorem follows by combining the above terms. 
\end{proof} 
A secure PRF and AEAD exhibiting negligible advantages, along with a good estimation $\hat{\pi}$ of the true query distribution $\pi$ such that $\advantage{dist}{\adv, \pi,\hat{\pi}}$ is negligible, will restrict the advantage to be negligible for any PPT adversary in the ROR-CRDA game.

\begin{theorem}[DO$_{\text{update}}$-ODDS]Let $\varepsilon=\bigO{1}$, $k=\bigO{1}$, $\db$ and $\db'$ be two neighboring data stores, $\ops$ and $\ops'$ be two query-consistent neighboring operational sequences of, $n=|\ops|$, and $\log^*|\domain|\leq\log\log\secpar$.
Then \system satisfies $(\varepsilon\log k , \delta)$-DO$_{\text{update}}$-ODDS and achieves perfect correctness.
\end{theorem}
\begin{proof}
The perfect correctness can be deduced using the same argument as presented in~\cite{ChanCMS19} as the iterative $k$-way merge algorithm does not introduce any correctness loss. 
By observing that every element is involved in only $\log k$ instances of the $(\varepsilon, \delta)$-DO merge algorithm~\citep{ChanCMS19}, the theorem follows straightforwardly from the $\log k$-fold composition rule of differential privacy~\citep{Dwork2010BoostingAD}.
\end{proof}

\end{document}